\documentclass[11pt,a4paper]{article}
\usepackage[margin=2cm]{geometry}
\usepackage[T1]{fontenc}
\usepackage[utf8]{inputenc}
\usepackage[noadjust]{cite}
\usepackage{lmodern}
\usepackage{amsthm}
\usepackage{thmtools}
\makeatletter
\@ifundefined{newcounteralias}{}{%
  \renewcommand\thmt@autorefsetup{%
    \@xa\def\csname\thmt@envname autorefname\@xa\endcsname
      \@xa{\thmt@thmname}%
  }%
}
\makeatother

\usepackage{thm-autoref}
\usepackage[shortlabels]{enumitem}
\usepackage[center]{caption}
\usepackage{latexsym}
\usepackage{amsfonts}
\usepackage{graphicx}
\usepackage{mathrsfs}
\usepackage[normalem]{ulem}
\usepackage{amssymb, mathdots, mathtools, xkeyval, tikz, tkz-euclide, adjustbox, float, array, subcaption, listings}
\usepackage{comment}
\usepackage{xspace}
\usepackage[colorlinks=true,
linkcolor =blue,
citecolor=blue,
urlcolor=blue]{hyperref}
\usepackage[capitalise]{cleveref}

\theoremstyle{plain}

\newtheorem{theorem}{Theorem}[section]
\newtheorem{lemma}[theorem]{Lemma}

\newtheorem{corollary}[theorem]{Corollary}

\newtheorem{claim}{Claim}[theorem]
\newenvironment{claimproof}{\noindent\textit{Proof of Claim \theclaim:}}{\hfill$\lrcorner$\\} 

\crefname{claim}{Claim}{Claims}

\newcommand{\col}[1]{\textsc{#1-Coloring}\xspace}
\newcommand{\lcol}[1]{\textsc{List #1-Coloring}\xspace}
\newcommand{\Wone}{\textsf{W[1]}\xspace}
\newcommand{\NP}{\textsf{NP}\xspace}
\newcommand{\XP}{\textsf{XP}\xspace}
\newcommand{\FPT}{\textsf{FPT}\xspace}
\newcommand{\Clique}{\textsc{Multicolored Clique}\xspace}
\newcommand{\MSI}{\textsc{Multicolored Subgraph Isomorphism}\xspace}
\newcommand{\Oh}{\mathcal{O}}

\newcommand{\port}{\textsf{port}}
\renewcommand{\phi}{\varphi}

\newcommand{\cC}{\mathcal{C}}

\newcommand{\cF}{\mathcal{F}}

\newcommand{\cS}{\mathcal{S}}

\usepackage{todonotes}

\title{Parameterized complexity of $k$-\textsc{Coloring}\\
in graphs with no long induced paths}

\author{Paweł Rzążewski\thanks{Warsaw University of Technology. Supported by the National Science Centre grant 2024/54/E/ST6/00094.}}
\date{}

\begin{document}
	\begin{titlepage}
	\date{}
	\maketitle

\begin{abstract}
    We study the parameterized complexity of \textsc{(List) $k$-Coloring} in $H$-free graphs, where $H$ is a linear forest, that is, a disjoint union of paths.
    First, considering $k$ as the parameter, we establish the following:
    \begin{enumerate}
      \item For any $s \geq 0$, \textsc{List $k$-Coloring} in $(P_4+sP_1)$-free graphs is fixed-parameter tractable (\textsf{FPT}).
      \item \textsc{$k$-Coloring} is \textsf{W[1]}-hard in $2P_2$-free graphs.
    \end{enumerate}
    The second result settles, in a strong form, a long-standing open problem posed by Hoàng, Kamiński, Lozin, Sawada, and Shu [Algorithmica, 2010].
    
    Next, we prove that \textsc{$k$-Coloring} is \textsf{NP}-hard in $(P_4+P_2)$-free graphs. These three findings, together with known results from classical, non-parameterized complexity, yield a complete complexity classification of \textsc{$k$-Coloring} and \textsc{List $k$-Coloring} in $H$-free graphs, parameterized by $k$, into the cases: \textsf{FPT}, \textsf{XP} but \textsf{W[1]}-hard, and \textsf{paraNP}-hard.

    We also prove that \textsc{$3$-Coloring} is \textsf{W[1]}-hard in $P_t$-free graphs when parameterized by $t$. This answers a question of Golovach, Johnson, Paulusma, and Song [Journal of Graph Theory, 2017].

    Finally, as a byproduct of our algorithm for \textsc{List $k$-Coloring} in $(P_4+sP_1)$-free graphs,
    we show that, for every fixed $s$ and $k$, there are only finitely many $(P_4+sP_1)$-free minimal obstructions to $k$-colorability. This settles a conjecture of Cameron, Hoàng, and Sawada [Disc. Appl. Math., 2022] and completes the dichotomy concerning the finiteness of the family of vertex-$k$-critical $H$-free graphs for every graph $H$ and every $k$.
\end{abstract}

	\thispagestyle{empty}

  \bigskip
  \bigskip
  \bigskip
\tableofcontents
\end{titlepage}
\thispagestyle{empty}
\newpage\setcounter{page}{1}

\section{Introduction}
The complexity of \col{$k$} in restricted graph classes is one of the central topics in algorithmic graph theory.
Particular attention has been paid to \emph{$H$-free graphs}, i.e., graphs that do not contain a fixed graph $H$ as an induced subgraph~\cite{DBLP:journals/gc/RanderathS04,DBLP:journals/jgt/GolovachJPS17}.
For an integer $t$, by $P_t$ we denote the path on $t$ vertices.
For graphs $H_1$ and $H_2$, by $H_1 + H_2$ we denote the disjoint union of $H_1$ and $H_2$.
We write $sH$ for the disjoint union of $s$ copies of $H$.

Classical hardness results imply that whenever $H$ is not a linear forest, \col{$k$} is \NP-hard on $H$-free graphs for every fixed $k \ge 3$~\cite{DBLP:journals/cpc/Emden-WeinertHK98,DBLP:journals/siamcomp/Holyer81a,DBLP:journals/jal/LevenG83}.
By contrast, for every fixed $k$, \col{$k$} is polynomial-time solvable on $H$-free graphs if either (a) $H$ is an induced subgraph of $sP_3$ for some $s$~\cite{DBLP:journals/corr/abs-2505-00412,DBLP:journals/combinatorica/ChudnovskyHS24,DBLP:journals/siamdm/HajebiLS22}, or (b) $H$ is an induced subgraph of $P_5 + sP_1$ for some $s$~\cite{DBLP:journals/algorithmica/0001GKP15,DBLP:journals/algorithmica/HoangKLSS10}.
These algorithmic results actually work for the more general \lcol{$k$} problem, where each vertex $v$ is assigned a list $L(v) \subseteq [k]$ of admissible colors, and we ask for a proper coloring respecting these lists.

For $k \ge 5$, \lcol{$k$} in $H$-free graphs is \NP-hard for every linear forest not covered by the tractable cases listed above~\cite{DBLP:journals/iandc/GolovachPS14,DBLP:journals/ejc/Huang16,DBLP:journals/algorithmica/0001GKP15}.
For \col{$4$}, the problem is polynomial-time solvable in $P_6$-free graphs~\cite{DBLP:conf/soda/SpirklCZ19} and in $H$-free graphs whenever $H$ is a linear forest on at most 5 vertices~\cite{GolovachPaulusmaSong2013Coloring} (this also follows from the results for general $k$ listed above).
On the other hand, the problem is \NP-hard in $P_7$-free graphs~\cite{DBLP:journals/ejc/Huang16}.

The picture for \col{$3$} is much less clear. Polynomial-time algorithms are known for $P_7$-free graphs~\cite{DBLP:journals/combinatorica/BonomoCMSSZ18} and for $(P_6+sP_3)$-free graphs for every fixed $s$~\cite{DBLP:journals/algorithmica/ChudnovskyHSZ21}. It is believed that \col{$3$} is polynomial-time solvable in $H$-free graphs for every linear forest $H$. The existence of a quasipolynomial-time algorithm for all such cases~\cite{DBLP:conf/sosa/PilipczukPR21} lends further support to this view and suggests that the problem is unlikely to be \NP-hard. At the same time, we still seem far from obtaining a polynomial-time algorithm.

In this paper, we focus on the parameterized complexity of \col{$k$} in $H$-free graphs. Among possible choices of the parameter, arguably the most natural one is $k$ itself.
Restating the results listed above, for every fixed $s$, the \col{$k$} problem parameterized by $k$ belongs to the class \XP when restricted to $sP_3$-free graphs or $(P_5 + sP_1)$-free graphs,
while it is \textsf{paraNP}-hard in $H$-free graphs whenever $H$ is not a linear forest or contains $P_6$ as an induced subgraph.
It is also known that the problem is \FPT in $H$-free graphs if $H = P_4$, as shown by Jansen and Scheffler~\cite{DBLP:journals/dam/JansenS97}, and if $H \in \{P_3 + P_1\} \cup \bigcup_{s \geq 1}\{P_2+sP_1\}$, as shown by Couturier, Golovach, Kratsch, and Paulusma~\cite{DBLP:journals/jda/CouturierGKP12}.

\paragraph{Our results.}
As our first result, we significantly extend the algorithmic results discussed above by showing the following result.

\begin{restatable}{theorem}{thmpfour}
    \label{thm:p4sp1algo}
    For every fixed $s \geq 0$, \lcol{$k$} in $n$-vertex $(P_4+sP_1)$-free graphs can be solved in time $2^{\Oh(k^5)} \cdot n^{\Oh(1)}$.
    In particular, the problem is \FPT when parameterized by $k$.
\end{restatable}

Our second result originates in the paper of Hoàng, Kamiński, Lozin, Sawada, and Shu~\cite{DBLP:journals/algorithmica/HoangKLSS10} showing that for every fixed $k$, \col{$k$} can be solved in polynomial time in $P_5$-free graphs. The authors asked whether the problem is \FPT (when parameterized by $k$).
Later, Couturier, Golovach, Kratsch, and Paulusma~\cite{DBLP:journals/jda/CouturierGKP12} observed that the question remains open even for $2P_2$-free graphs.
This question was later repeated in subsequent papers and at numerous workshops~\cite{DBLP:journals/jgt/GolovachJPS17,DBLP:journals/siamdm/ChudnovskyKPRS21,DBLP:journals/dagstuhl-reports/ChudnovskyPS19,DBLP:journals/dagstuhl-reports/FellowsGMS12,DBLP:journals/dagstuhl-reports/ChudnovskyMPSA22}. As noted by Golovach~\cite{DBLP:journals/dagstuhl-reports/ChudnovskyMPSA22}, the problem ``remains open for a long time despite all efforts.''
Quite surprisingly, we show that the problem is actually \Wone-hard.

\begin{restatable}{theorem}{thmtwoptwo}
    \label{thm:2p2hard}
    \col{$k$} in $2P_2$-free graphs is \Wone-hard when parameterized by $k$.
    Furthermore, unless the ETH fails, the problem cannot be solved in time $f(k) \cdot n^{o(k)}$ on $n$-vertex instances for any computable function $f$.
\end{restatable}

We remark that this lower bound matches the running time of the algorithm of Chudnovsky, King, Pilipczuk, Rzążewski, and Spirkl~\cite{DBLP:journals/siamdm/ChudnovskyKPRS21} for a superclass of $2P_2$-free graphs.
Indeed, they showed that in $n$-vertex $P_5$-free graphs, \lcol{$k$} (and even some generalizations of the problem) can be solved in time $n^{\Oh(k)}$.

\cref{thm:p4sp1algo,thm:2p2hard}, together with polynomial-time algorithms~\cite{DBLP:journals/corr/abs-2505-00412,DBLP:journals/combinatorica/ChudnovskyHS24,DBLP:journals/siamdm/HajebiLS22,DBLP:journals/algorithmica/0001GKP15,DBLP:journals/algorithmica/HoangKLSS10} and \NP-hardness results~\cite{DBLP:journals/iandc/GolovachPS14,DBLP:journals/ejc/Huang16,DBLP:journals/algorithmica/0001GKP15} listed above, yield a complete picture of the parameterized complexity of \lcol{$k$} in $H$-free graphs when parameterized by $k$.

Actually, combining all these results \emph{almost} gives the same dichotomy for \col{$k$}.
The only missing ingredient concerns $(P_4+P_2)$-free graphs: while \lcol{$5$} is known to be \NP-hard in this class~\cite{DBLP:journals/algorithmica/0001GKP15},
no hardness result was known for \col{$k$}, for any constant $k$.
We fill this gap by proving the following result.

\begin{restatable}{theorem}{thmsixcoloring}
    \label{thm:p4p2hard}
    For any $k \geq 6$, \col{$k$} is \NP-hard in $(P_4+P_2)$-free graphs.
\end{restatable}

While the complexity of \col{$5$} in $(P_4+P_2)$-free graphs remains open, \cref{thm:p4p2hard} is sufficient to obtain the following dichotomy.

\begin{restatable}{theorem}{thmdicho}
    \label{thm:parameterized-dichotomy}
    Let $H$ be a fixed graph on at least two vertices.
    The \col{$k$} and \lcol{$k$} problems in $H$-free graphs, when parameterized by $k$, are
    \begin{itemize}
        \item \FPT if $H$ is an induced subgraph of $P_4 + sP_1$ for some $s \geq 0$;
        \item \XP and \Wone-hard if $H$ contains $2P_2$ as an induced subgraph and is an induced subgraph of $P_5 + sP_1$ or $sP_3$ for some $s \geq 0$;
        \item \textsf{paraNP}-hard otherwise.
    \end{itemize}
\end{restatable}

Next, we continue studying the parameterized complexity of \col{$k$} in $P_t$-free graphs, but we introduce a change in the setting:
we fix the number of colors to three and instead parameterize by the length $t$ of the forbidden induced path.
Recall that it is plausible that for every fixed $t$, \col{$3$} can be solved in polynomial time in $P_t$-free graphs. However, known algorithms for small cases become increasingly complex as $t$ grows~\cite{DBLP:journals/dm/RanderathST02,DBLP:journals/dam/RanderathS04,DBLP:journals/combinatorica/BonomoCMSSZ18}.
Thus, it is natural to ask whether \col{$3$} in $P_t$-free graphs is \Wone-hard when parameterized by $t$.
This question was stated by Golovach, Johnson, Paulusma, and Song~\cite{DBLP:journals/jgt/GolovachJPS17}; see also~\cite{DBLP:journals/dagstuhl-reports/ChudnovskyPS19}.
As our third result, we answer this question in the affirmative.

\begin{restatable}{theorem}{thmpt}
    \label{thm:pthard}
    \col{$3$} in $P_t$-free graphs is \Wone-hard when parameterized by $t$.
    Furthermore, unless the ETH fails, the problem cannot be solved in time $f(t) \cdot n^{o(t / \log t)}$ on $n$-vertex instances for any computable function $f$.
\end{restatable}

Thus, even if it turns out that for every fixed $t$, \col{$3$} can be solved in polynomial time in $P_t$-free graphs,
\cref{thm:pthard} rules out a uniform algorithm with running time $f(t) \cdot n^{\Oh(1)}$ for any computable function $f$ (unless \FPT = \Wone).

\bigskip
Finally, we somewhat diverge from the main topic of the paper and observe that our approach from the proof of \cref{thm:p4sp1algo} has an interesting graph-theoretic consequence.
For $k \geq 3$, we say that a graph is \emph{vertex-$(k+1)$-critical} if it is not $k$-colorable, but every proper induced subgraph of it is $k$-colorable.
An active line of research in graph theory is to study the structure of vertex-$k$-critical graphs in various graph classes~\cite{MaffrayMorel2012,GoedgebeurSchaudt2018,DBLP:journals/siamdm/ChudnovskyGSZ20,CameronEtAl2021,CaiGoedgebeurHuang2023,BeatonCameron2025,BelavadiKarthick2026,BeatonCameron2026Bipartite,ChudnovskyEtAl2020P6}.
In particular, if there are only finitely many vertex-$k$-critical graphs in a graph class, then the \col{$k$} problem can be solved in polynomial time in this class. Furthermore, such an algorithm is \emph{certifying}, i.e., in case of a no-instance, it can output a small certificate of non-$k$-colorability.

Full classification of pairs $(k,H)$ for which there are only finitely many vertex-$k$-critical $H$-free graphs is known only for $k \leq 4$.
Furthermore, it is known for $k \geq 5$ that there are infinitely many vertex-$k$-critical $H$-free graphs whenever $H$ is not an induced subgraph of $P_4 + sP_1$ for some $s \geq 0$~\cite{Erdos1959,HoangEtAl2015,DBLP:journals/dam/CameronHS22}.
It was conjectured by Cameron, Hoàng, and Sawada~\cite{DBLP:journals/dam/CameronHS22} that for every fixed $s$ and $k$, there are only finitely many vertex-$k$-critical $(P_4+sP_1)$-free graphs.
Confirming this conjecture would complete the dichotomy for finiteness of the family of vertex-$k$-critical $H$-free graphs.
This problem was later studied by several authors, but only partial results were obtained~\cite{DBLP:journals/dam/CameronHS22,AbuadasEtAl2024,BeatonCameron2025,BeatonCameron2026Subfamilies,BelavadiKarthick2026,BeatonCameron2026Bipartite}.
Using our approach, we can confirm the conjecture in its full generality.

\begin{restatable}{theorem}{thmobstructions}
    \label{thm:obstructions}
    For every fixed $s$ and $k$, there are finitely many vertex-$k$-critical $(P_4+sP_1)$-free graphs.
\end{restatable}

Combining this with known results~\cite{Erdos1959,HoangEtAl2015,DBLP:journals/siamdm/ChudnovskyGSZ20,DBLP:journals/dam/CameronHS22,ChudnovskyEtAl2020P6}, we obtain the following dichotomy.

\begin{theorem}\label{thm:obstructions-dichotomy}
    Let $H$ be a fixed graph on at least two vertices.
    For $k\geq 3$, the family of vertex-$k$-critical $H$-free graphs is finite if and only if
    \begin{enumerate}
        \item $k=3$ and $H$ is a linear forest, or
        \item $k=4$ and $H$ is an induced subgraph of a graph in $\{P_6, 2P_3\} \cup \bigcup_{s \geq 0} \{P_4 + sP_1\}$, or
        \item $k \geq 5$ and $H$ is an induced subgraph of a graph in $\bigcup_{s \geq 0} \{P_4 + sP_1\}$.
    \end{enumerate}
\end{theorem}

\paragraph{Our techniques.}
The algorithm of \cref{thm:p4sp1algo} is an exhaustive application of a single reduction rule: as long as the graph is large, we find a nonempty set $Z \subseteq V(G)$ such that $(G,L)$ and $(G-Z,L)$ are equivalent instances of \lcol{$k$}, delete $Z$, and repeat. Once the rule no longer applies, the graph has at most $2^{\Oh(k^2)}$ vertices, and we can afford to solve such a small instance, e.g., by an \XP algorithm.
The size threshold comes from the independence number: a $k$-colorable graph in which every independent set has fewer than $M = 2^{\Oh(k^2)}$ vertices has fewer than $kM$ vertices, so an instance on more than $kM$ vertices and with no independent set of size $M$ can be rejected right away. Thus the interesting case is that we are handed a large independent set $X$, and all the work goes into converting $X$ into a deletable set $Z$.
First, a Ramsey-type lemma allows us to either extract a clique of size $k+1$ -- in which case we reject the instance -- or split $V(G)$ into two sets $W$ and $D$, where $G[W]$ is a \emph{cograph}, i.e., is $P_4$-free.
In $G[W]$ we distinguish $q = 2^{\Oh(k)}$ components and argue that every vertex of $D$ is complete to all but at most $s$ of the distinguished components.

The heart of the proof is the \emph{reduction lemma} (\cref{lem:reduce}): we show that most of the distinguished components are actually irrelevant and can be deleted.
For every set $K \subseteq [k]$ of colors we keep $sk+1$ distinguished components that are not list colorable using only colors from $K$; as $q > (sk+1)2^k$, the union $Z$ of the remaining ones is nonempty, and we claim that it can be safely deleted. Indeed, consider a coloring $\phi$ of $G-Z$ and let $K$ be the set of colors unused on $D$. As every vertex of $D$ is complete to all but $s$ distinguished components, at most $sk$ surviving components use a color outside $K$, so fewer than $sk+1$ of them were kept for $K$. Thus every deleted component can be colored using only the colors from $K$, and such colorings extend $\phi$ to $G$.

Both proofs of \cref{thm:2p2hard,thm:pthard} are based on reductions from the well-known \MSI problem~\cite{DBLP:journals/toc/Marx10}, but with different pattern graphs.
We design two different \emph{selector gadgets} whose possible colorings encode a choice of a vertex in the instance of \MSI.

For \cref{thm:2p2hard}, the pattern graph is a clique (so, in other words, we reduce from \Clique).
We first construct a hard instance of \lcol{$k$}. The construction is carefully designed to remain $2P_2$-free and to satisfy some additional properties that allow us to eliminate the lists using the standard reduction to \col{$k$}.

For \cref{thm:pthard}, the pattern graph is assumed to be subcubic.
We design a selector gadget that excludes a fixed path as an induced subgraph,
and combine many copies of it to encode an instance of \MSI.
To control the length of induced paths in the resulting graph,
we use an auxiliary observation bounding it in terms of the induced-matching sizes between the gadgets.

The proof of \cref{thm:p4p2hard} is much simpler: we reduce from $3$-\textsc{Precoloring Extension} in bipartite graphs.
After some easy preprocessing we may assume that each side of the bipartition contains exactly three precolored vertices, one of each color.
These six vertices are turned into a clique, which splits the six colors into pairs $\{i,i'\}$, one for each color $i$ of the original instance,
and every edge $e$ is encoded by three vertices $z_{1,e},z_{2,e},z_{3,e}$, where $z_{i,e}$ can use only the colors $i$ and $i'$ and is adjacent to both endpoints of $e$.

Finally, \cref{thm:obstructions} is an immediate consequence of the reduction rule: applied to a graph that is not $k$-colorable, it either returns a proper induced subgraph that is still not $k$-colorable, i.e., a smaller obstruction, or reports that the graph is already small. Hence every minimal obstruction has at most $2^{\Oh(k^2)}$ vertices.

\newpage

\section{Preliminaries}
For a positive integer $n$, we write $[n]=\{1,\ldots,n\}$.

\paragraph{Graph theory.}
Let $G$ be a graph. By $V(G)$ and $E(G)$ we denote, respectively, the vertex set and the edge set of $G$.
For a vertex $v$, by $N(v)$ we denote the set of neighbors of $v$, and we write $N[v] = N(v) \cup \{v\}$.
For a set $X \subseteq V(G)$, we write $N[X] = \bigcup_{v \in X} N[v]$ and $N(X) = N[X] \setminus X$.
A graph is \emph{subcubic} if every vertex has degree at most three.

For a set $X \subseteq V(G)$, by $G[X]$ we denote the subgraph of $G$ induced by $X$,
and by $G-X$ we denote the graph $G[V(G) \setminus X]$.

Let $A,B$ be disjoint sets of vertices of a graph $G$, and let $v$ be a vertex.
We say that they are \emph{complete} (resp., \emph{anticomplete}) to each other if all edges (resp., no edges) between $A$ and $B$ exist.
In particular, a vertex $v$ is complete (resp., anticomplete) to a set $A$ if $\{v\}$ and $A$ are complete (resp., anticomplete) to each other.
A vertex $v$ is \emph{mixed} on a set $A$ if it is neither complete nor anticomplete to $A$.

By $E(A,B)$ we denote the set of edges with one endpoint in $A$ and the other in $B$. By $G[A,B]$ we denote the bipartite graph with vertex set $A \cup B$ and edge set $E(A,B)$.

A \emph{chain graph} is a bipartite graph with parts $X=\{x_1,\ldots,x_n\}$ and $Y=\{y_1,\ldots,y_n\}$ such that $x_i$ is adjacent to $y_j$ if and only if $i \leq j$.
It is straightforward to verify that a chain graph is $2P_2$-free.

\paragraph{Parameterized complexity, \Wone-hardness, and ETH.}
For all notions related to parameterized complexity, we refer the reader to the textbook of Cygan et al.~\cite{DBLP:books/sp/CyganFKLMPPS15}. We only recall the two source problems used in our reductions.

An instance of \MSI is a pair $(G,H)$ of graphs, where $G$ is called \emph{host} and $H$ is called \emph{pattern}.
Let $r$ be the number of vertices of $H$ and let $V(H)=[r]$.
The vertex set of $G$ is partitioned into $r$ sets $V^1,\ldots,V^r$.
The goal is to determine whether there is $v_i \in V^i$ for every $i \in [r]$,
so that, for every $ij \in E(H)$ we have $v_iv_j \in E(G)$.
Without loss of generality we may assume that each set $V^i$ is independent, all these sets are of equal size, and edges between $V^i$ and $V^j$ exist if and only if $ij \in E(H)$.
The \Clique problem is a special case of \MSI where $H$ is a clique.
For brevity, we will denote an instance of \Clique by $(G,r)$ instead of $(G,K_r)$.

The following theorem summarizes known hardness results for \MSI and \Clique.

\begin{theorem}[\cite{DBLP:journals/toc/Marx10,DBLP:conf/stacs/MarxP14,DBLP:conf/iwpec/EppsteinL18,DBLP:books/sp/CyganFKLMPPS15}]\label{thm:msi}
    Let $f$ be any computable function.
    \MSI on instances $(G,H)$ with $n$-vertex host graph $G$ and $r$-vertex pattern graph $H$:
    \begin{enumerate}
        \item is \Wone-hard and cannot be solved in time $f(r) \cdot n^{o(r)}$ if $H$ is a clique, unless the ETH fails.
        \item is \Wone-hard and cannot be solved in time $f(r) \cdot n^{o(r/\log r)}$ if $H$ is subcubic, unless the ETH fails.
    \end{enumerate}
\end{theorem}

\newpage



\section{\lcol{$k$} in $(P_4+sP_1)$-free graphs}
In this section we prove \cref{thm:p4sp1algo}.
Before we proceed to the proof, we introduce some basic tools.

\paragraph{Basic tools.} The following property is well-known and easy to observe~\cite{DBLP:journals/siamdm/ChudnovskyKPRS21}.

\begin{lemma}\label{lem:mixed}
    Let $G$ be a graph, let $C$ be a connected set and let $v \notin C$ be mixed on $C$.
    Then there are two adjacent vertices $x,y \in C$ such that $xv \in E(G)$ and $yv \notin E(G)$.
\end{lemma}

A $P_4$-free graph is also known as a \emph{cograph}. Equivalently, cographs are exactly the graphs that can be constructed from single-vertex graphs by repeatedly applying the operations of disjoint union and join. In particular, every cograph with at least two vertices is either a disjoint union or a join of two cographs, and in the former case it is disconnected. From this recursive definition, we immediately obtain the following property of cographs.

\begin{lemma}\label{lem:dominator-cographs}
    Let $G$ be a connected cograph and let $I$ be an independent set in $G$ with $|I| \geq 2$.
    Then there is a vertex $v \in V(G) \setminus I$ that is complete to $I$.
\end{lemma}
\begin{proof}
    As $G$ is connected and has at least two vertices, it is the join of two nonempty cographs $G_1$ and $G_2$.
    As $I$ is independent and $V(G_1)$ is complete to $V(G_2)$, the set $I$ cannot intersect both $V(G_1)$ and $V(G_2)$; by symmetry assume that $I \subseteq V(G_1)$.
    Then every $v \in V(G_2)$ is complete to $I$ and satisfies $v \notin I$.
\end{proof}

Cographs are equivalently graphs of clique-width at most 2, and thus many problems can be solved efficiently on cographs using dynamic programming on a decomposition of the graph. In particular, we have the following result.

\begin{theorem}[Jansen and Scheffler~\cite{DBLP:journals/dam/JansenS97}]\label{thm:lcol-cographs}
    For every $k \geq 1$, \lcol{$k$} in $n$-vertex $P_4$-free graphs can be solved in time $2^{\Oh(k)} \cdot n^{\Oh(1)}$.
\end{theorem}

The next tool is the following result on finding a largest independent set in $(P_4+sP_1)$-free graphs.

\begin{theorem}\label{thm:mis-p4sp1}
    For every fixed $s$, a largest independent set in an $n$-vertex $(P_4+sP_1)$-free graph can be found in polynomial time.
\end{theorem}
\begin{proof}
    Let $G$ be the input graph.
    First, enumerate all independent sets of size at most $s$; this can be done in polynomial time, as $s$ is a constant.
    If no independent set of size $s$ is found, then we return a largest set found in this phase.

    Otherwise, for every independent set $I$ of size $s$, delete $I$ and all its neighbors from $G$, and let $G'$ be the remaining graph.
    Note that any induced $P_4$ in $G'$, together with $I$, induces a $P_4 + sP_1$ in $G$. Thus, $G'$ is $P_4$-free, i.e., a cograph.
    Now, a largest independent set in $G'$ can be found in polynomial time, e.g., by exploiting the boundedness of clique-width~\cite{DBLP:journals/mst/CourcelleMR00}.
    Among all choices of $I$, we return a largest set of the form $I \cup J$, where $J$ is a largest independent set of the corresponding graph $G'$; note that every such set is independent in $G$.
    This is optimal: if $I^\star$ is a largest independent set of $G$, then $|I^\star| \geq s$, so we may pick $I \subseteq I^\star$ with $|I|=s$, and then $I^\star \setminus I$ is an independent set of the corresponding graph $G'$.
\end{proof}

Finally, we will use the following result.
It can be obtained using the approach of Couturier, Golovach, Kratsch, and Paulusma~\cite{DBLP:journals/algorithmica/0001GKP15} for $(P_5+sP_1)$-free graphs
combined with the $n^{\Oh(k)}$-algorithm of Chudnovsky, King, Pilipczuk, Rzążewski, and Spirkl~\cite{DBLP:journals/siamdm/ChudnovskyKPRS21} for $P_5$-free graphs.

\begin{theorem}\label{thm:lcol-p5sp1}
    For every fixed $s$, \lcol{$k$} in $n$-vertex $(P_5+sP_1)$-free graphs can be solved in time $n^{\Oh(k^3)}$.
\end{theorem}

Let us remark that, as we only need to apply \cref{thm:lcol-p5sp1} to $(P_4+sP_1)$-free graphs, we are quite confident that the running time can be improved.
This would immediately reduce the overall running time in \cref{thm:p4sp1algo}. However, we do not pursue this direction here, as it is not the main focus of this paper.

\paragraph{Constants and conventions.}
Throughout this section we fix two integers $s,k \geq 1$.
We treat $s$ as a constant: all constants hidden in the $\Oh(\cdot)$-notation, as well as the exponents of the polynomial factors in the running times, are allowed to depend on $s$.
On the other hand, $k$ is the parameter, so we keep track of the dependence on $k$ explicitly.

For future reference, we gather in one place all the constants and functions used in this section:
\begin{equation}\label{eq:constants}
    \begin{aligned}
        h_s(r,p) &:= p\left(1+rp^s\right)^r && \text{for integers } r \geq 0 \text{ and } p > s,\\
        q &:= (sk+1)2^k+1,\\
        \sigma &:= s \cdot q = s\left((sk+1)2^k+1\right),\\
        M &:= h_s(2k,\sigma).
    \end{aligned}
\end{equation}

We will use the following three properties of these constants; each of them is immediate from \eqref{eq:constants} and $s,k \geq 1$:
\begin{equation}\label{eq:constants-prop}
    \sigma > s, \qquad q \geq 2(s+1)+1 = 2s+3, \qquad\text{and}\qquad q > (sk+1)2^k.
\end{equation}
Let us also estimate $M$.
As $s$ is a constant, we have $\sigma = 2^{k+\Oh(\log k)}$ and thus $\sigma^s = 2^{sk+\Oh(\log k)}$, which yields
\begin{equation}\label{eq:M-bound}
    M = \sigma\left(1+2k \cdot \sigma^s\right)^{2k} = 2^{2sk^2+\Oh(k \log k)} = 2^{\Oh(k^2)},
\end{equation}
where, as everywhere in this section, the constants hidden in the $\Oh(\cdot)$-notation depend on $s$.

\paragraph{Extracting an independent set with restricted neighborhood.}
Now, let us show that in a $(P_4+sP_1)$-free graph, if we are given a sufficiently large independent set, then we can extract a smaller independent set $S$ such that every vertex of $G$ is adjacent to almost all or to almost no vertices from $S$.

We start with the following combinatorial lemma concerning set systems.

\begin{lemma}\label{lem:extraction}
    Let $r \geq 0$ and $p > s$ be integers, and recall the function $h_s(\cdot,\cdot)$ defined in \eqref{eq:constants}.
    Let $X$ be a ground set and let $\cF \subseteq 2^X$ be a family of pairwise distinct subsets of $X$ such that every $s$-element subset of $X$ is contained in at most $r$ sets in $\cF$.
    If $|X| \geq h_s(r,p)$, then there exists a set $S \subseteq X$ of size $p$ such that, for every $F \in \cF$,
    \begin{equation}
        S \subseteq F \quad\text{or}\quad |S \cap F| \leq s. \label{eq:extract}
    \end{equation}
    Moreover, such a set $S$ can be found in time polynomial in $|X| + |\cF|$.
\end{lemma}
\begin{proof}
    We proceed by induction on $r$; for a fixed $r$, the statement is claimed for all ground sets $X$ and all families $\cF$ as above.

    If $r=0$, then no $s$-element subset of $X$ is contained in any set in $\cF$, and thus every set $F \in \cF$ has fewer than $s$ elements:
        indeed, if we had $|F| \geq s$ for some $F \in \cF$, then any $s$-element subset of $F$ would be an $s$-element subset of $X$ contained in $F$.
    As $|X| \geq h_s(0,p) = p$, we can pick a set $S \subseteq X$ of size $p$, and it satisfies \eqref{eq:extract} since $|S \cap F| \leq |F| < s$ for every $F \in \cF$.

    Now, suppose that $r \geq 1$ and the statement holds for $r-1$. Define
    \[
        u := p(1+rp^s)^{r-1}.
    \]
    Note that $u \geq p$ and $h_s(r,p) = u(1+rp^s)$.
    Furthermore, $h_s(r-1,p) = p(1+(r-1)p^s)^{r-1} \leq p(1+rp^s)^{r-1} = u$.

    Suppose first that some $Y \in \cF$ is of size at least $u$.
    We restrict the ground set to $Y$ and define
    \[
        \cF' = \{ F \cap Y \mid F \in \cF \setminus \{Y\} \}.
    \]
    Note that $\cF' \subseteq 2^{Y}$; we again treat $\cF'$ as a set, so if two distinct sets $F,F'' \in \cF \setminus \{Y\}$ satisfy $F \cap Y = F'' \cap Y$, then this common set is kept in $\cF'$ only once.
    We claim that every $s$-element subset $T$ of $Y$ is contained in at most $r-1$ sets in $\cF'$.
    Indeed, for every $F' \in \cF'$ containing $T$, let us fix some $F \in \cF \setminus \{Y\}$ with $F' = F \cap Y$; as $T \subseteq Y$, we have $T \subseteq F$.
    Distinct sets in $\cF'$ yield distinct sets $F$, so the number of sets in $\cF'$ containing $T$ is at most the number of sets in $\cF \setminus \{Y\}$ containing $T$.
    The latter is at most $r-1$, because $T$ is contained in at most $r$ sets in $\cF$, and one of them is $Y$.

    Since $|Y| \geq u \geq h_s(r-1,p)$, we may apply the induction hypothesis to the ground set $Y$ and the family $\cF'$.
    We obtain a set $S \subseteq Y \subseteq X$ of size $p$ such that, for every $F' \in \cF'$, either $S \subseteq F'$ or $|S \cap F'| \leq s$.
    Take any $F \in \cF$.
    If $F = Y$, then $S \subseteq F$.
    Otherwise, let $F' = F \cap Y \in \cF'$.
    If $S \subseteq F'$, then again $S \subseteq F$.
    Finally, if $|S \cap F'| \leq s$, then $|S \cap F| = |S \cap F'| \leq s$ as $S \subseteq Y$.
    Thus, $S$ satisfies the conditions in the statement of the lemma.

    Now, suppose that every set in $\cF$ has size less than $u$.
    We construct a set $S$ as follows:
        we greedily pick elements from $X$, maintaining that we select at most $s$ elements from any $F \in \cF$.
    We finish when we have selected $p$ elements or when we cannot select any more elements; let $S$ be the set of selected elements.

    Suppose that $|S| < p$.
    Let $\cF_S$ be the family of sets in $\cF$ that contain exactly $s$ elements from $S$.
    Every $F \in \cF_S$ contains the $s$-element set $F \cap S \subseteq S$, and every $s$-element subset of $S$ is contained in at most $r$ sets in $\cF$, so, as $|S| < p$, we obtain
    \[
        |\cF_S| \leq r\binom{|S|}{s} \leq r\binom{p-1}{s} < rp^s.
    \]
    We observe that $S$ cannot be enlarged, so every element in $X$ is either already in $S$, or it belongs to some set in $\cF_S$.
    Indeed, consider $x \in X \setminus S$; as $x$ was not selected, there is $F \in \cF$ with $|F \cap (S \cup \{x\})| \geq s+1$.
    Since no set in $\cF$ contains $s+1$ elements of $S$, this means that $x \in F$ and $|F \cap S| = s$, i.e., $F \in \cF_S$.
    As each set in $\cF$, and hence each set in $\cF_S$, has fewer than $u$ elements, we obtain
    \[
       |X| \leq |S| + |\cF_S|u < p + urp^s \leq u(1+rp^s) = h_s(r,p) \leq |X|,
    \]
    where the last but one inequality follows from $u \geq p$; a contradiction.
    Consequently, $|S| = p$ and, as no set in $\cF$ contains $s+1$ elements of $S$, we have $|S \cap F| \leq s$ for every $F \in \cF$, so $S$ satisfies \eqref{eq:extract}.

    Finally, let us discuss the running time; note that the argument above is constructive.
    The recursion is linear, i.e., in each application of the inductive step we make at most one recursive call, and it is made only in the first case.
    In the inductive step, deciding which of the two cases applies and, in the first case, computing $\cF'$ (in particular, removing the duplicates) takes time polynomial in $|X| + |\cF|$.
    As $\cF'$ is obtained from $\cF \setminus \{Y\}$, we have $|\cF'| < |\cF|$, so the depth of the recursion is at most $\min(r,|\cF|) \leq |\cF|$.
    In the second case, which terminates the recursion, we perform at most $p \leq |X|$ selections, and for each of them we consider at most $|X|$ candidate elements, verifying in polynomial time whether a candidate can be selected.
    Thus the total running time is polynomial in $|X| + |\cF|$.
\end{proof}

Now, we apply \cref{lem:extraction} to the sufficiently large independent set in a $(P_4+sP_1)$-free graph to obtain the following lemma.

\begin{lemma}\label{lem:extract-is}
    Let $p > s$ be an integer.
    Given a $(P_4+sP_1)$-free graph $G$ and an independent set $X$ in $G$ of size at least $h_s(2k,p)$, in polynomial time we can find a clique of size $k+1$ in $G$,
    or a set $S \subseteq X$ of size $p$ such that, for every vertex $v \in V(G) \setminus S$,
    \begin{equation}
        |S \cap N(v)| \leq s \quad\text{or}\quad |S \setminus N(v)| \leq s. \label{eq:extract-is}
    \end{equation}
\end{lemma}
\begin{proof}
    For a vertex $v \notin X$, define $A_v = N(v) \cap X$, and
    \[
        F_v = \begin{cases}
            A_v & \text{ if } |A_v| \leq |X|/2, \\
            X \setminus A_v & \text{ if } |A_v| > |X|/2.
        \end{cases}
    \]
    Let $\cF = \bigcup_{v \in V(G) \setminus X} \{F_v\}$.
    We emphasize that $\cF$ is a set, i.e., it has no duplicates.

    \begin{claim}\label{clm:xminusauav}
    Let $u,v \in V(G) \setminus X$ be distinct vertices, such that $uv \notin E(G)$,
    $A_u \neq A_v$, and $A_u \cap A_v \neq \emptyset$.
    Then $|X \setminus (A_u \cup A_v)| < s$.
    \end{claim}
    \begin{claimproof}
    Suppose for contradiction that $|X \setminus (A_u \cup A_v)| \geq s$ and let $R$ be a set of $s$ vertices from $X \setminus (A_u \cup A_v)$.

    Since $A_u \neq A_v$, we have either $A_u \setminus A_v \neq \emptyset$ or $A_v \setminus A_u \neq \emptyset$. By symmetry, we may assume that there is $x \in A_u \setminus A_v$.
    Since $A_u \cap A_v \neq \emptyset$, there is $y \in A_u \cap A_v$.
    Note that $x,y \in X$ are distinct and nonadjacent, as $X$ is independent.
    Now, $xu,uy,yv \in E(G)$, while $xy,xv,uv \notin E(G)$, so $x,u,y,v$ induce a $P_4$ in $G$.
    Furthermore, every vertex of $R$ is nonadjacent to $x$ and to $y$, as $R \subseteq X$ and $X$ is independent, and it is nonadjacent to $u$ and to $v$, as $R \cap (A_u \cup A_v) = \emptyset$.
    Consequently, $x,u,y,v$ together with $R$ induce a $P_4+sP_1$ in $G$, a contradiction.
    \end{claimproof}

    Now, suppose there is an $s$-element subset $T \subseteq X$ contained in at least $2k+1$ sets in $\cF$.
    Picking one vertex per such set, we obtain a set $Y \subseteq V(G) \setminus X$ of size at least $2k+1$ such that $T \subseteq F_y$ for every $y \in Y$, and the sets $F_y$ for $y \in Y$ are pairwise distinct.
    By the pigeonhole principle, there is $Y' \subseteq Y$ of size at least $k+1$ such that either $F_y = A_y$ for every $y \in Y'$, or $F_y = X \setminus A_y$ for every $y \in Y'$.

    \begin{claim}\label{clm:yprime-clique}
    $Y'$ is a clique in $G$.
    \end{claim}
    \begin{claimproof}
    For contradiction, let $u,v \in Y'$ be two distinct nonadjacent vertices.
    Recall that the sets $F_y$ for $y \in Y'$ are pairwise distinct, so in particular $F_u \neq F_v$.

    Suppose first that $F_u = A_u$ and $F_v = A_v$, i.e., $|A_u| \leq |X|/2$ and $|A_v|\leq |X|/2$.
    Then $A_u \neq A_v$, as $F_u \neq F_v$.
    Furthermore, $T \subseteq A_u \cap A_v$, so in particular $|A_u \cap A_v| \geq s \geq 1$ and thus $A_u \cap A_v \neq \emptyset$.
    Consequently, we obtain
    \[
        |X \setminus (A_u \cup A_v) | = |X| - |A_u| - |A_v| + |A_u \cap A_v| \geq |X| - |X|/2 - |X|/2 + s = s.
    \]
    This contradicts \cref{clm:xminusauav}.

    Now, suppose that $F_u = X \setminus A_u$ and $F_v = X \setminus A_v$, i.e., $|A_u| > |X|/2$ and $|A_v| > |X|/2$.
    Again $A_u \neq A_v$, as $F_u \neq F_v$.
    Note that $A_u \cap A_v \neq \emptyset$ as each of these sets has size more than $|X|/2$.
    Then $T \subseteq X \setminus (A_u \cup A_v)$, so in particular $|X \setminus (A_u \cup A_v)| \geq s$. Again, this contradicts \cref{clm:xminusauav}.
    \end{claimproof}

    As $|Y'| \geq k+1$, \cref{clm:yprime-clique} provides us with a clique of size $k+1$ in $G$, which is one of the desired outputs.
    Thus, from now on let us assume that every $s$-element subset of $X$ is contained in at most $2k$ sets in $\cF$.
    Since $|X| \geq h_s(2k,p)$, we may apply \cref{lem:extraction} to the ground set $X$, family $\cF$, and $r=2k$.
    This gives us a set $S \subseteq X$ of size $p$ such that,
    for every $F \in \cF$, either $S \subseteq F$ or $|S \cap F| \leq s$.

    Let us show that $S$ satisfies \eqref{eq:extract-is} for every $v \in V(G) \setminus S$.
    For $v \in X \setminus S$ the condition holds trivially: as $X$ is independent and $S \subseteq X$, we have $|S \cap N(v)| = 0 \leq s$.
    So consider any $v \in V(G) \setminus X$.
    If $F_v = A_v$, then $S \subseteq F_v$ means that $S \subseteq N(v)$, so in particular $|S \setminus N(v)|=0\leq s$,
    and $|S \cap F_v| \leq s$ means that $|S \cap N(v)| \leq s$.
    If $F_v = X \setminus A_v$, then $S \subseteq F_v$ means that $|S \cap N(v)| = 0 \leq s$,  and $|S \cap F_v| \leq s$ means that $|S \setminus N(v)| \leq s$.
    Thus, in both cases, we have that $S$ satisfies \eqref{eq:extract-is}.

    Finally, we observe that the running time is polynomial in $n := |V(G)|$.
    First, we compute the sets $A_v$ and $F_v$ for all $v \in V(G) \setminus X$ and remove the duplicates, which takes time polynomial in $n$; note that $|\cF| \leq n$.
    Then, we exhaustively enumerate all sets $T$ of size $s$ in $X$ and, for each of them, count the sets of $\cF$ containing $T$; this takes time $\Oh(|X|^s) \cdot n^{\Oh(1)}$, which is polynomial in $n$ as $s$ is a constant.
    If some $T$ is contained in at least $2k+1$ sets of $\cF$, we find the set $Y'$ and return the clique, which again takes polynomial time.
    Otherwise, we call the algorithm from \cref{lem:extraction}, which runs in time polynomial in $|X| + |\cF| \leq 2n$.
\end{proof}

\paragraph{The structure around the extracted set.}
Now, let us analyze the structure of a $(P_4+sP_1)$-free graph $G$ around a set $S$ obtained from \cref{lem:extract-is}.

\begin{lemma}\label{lem:components}
    Let $q \geq 2s+3$ be an integer and let $\sigma := sq$.
    Let $G$ be a $(P_4+sP_1)$-free graph and let $S$ be an independent set of size at least $\sigma$ that, for every vertex $v \in V(G) \setminus S$, satisfies \eqref{eq:extract-is}.
    Define
    \[
        W = \{ x \in V(G) \mid |N(x) \cap S| \leq s \} \quad\text{and}\quad
        D = V(G) \setminus W.
    \]
    Then:
    \begin{enumerate}
        \item $G[W]$ is $P_4$-free,
        \item there are at least $q$ components of $G[W]$ that intersect $S$; let $C_1,\ldots,C_q$ be any $q$ of them, and
        \item every vertex of $D$ is complete to all but at most $s$ of the components $C_1,\ldots,C_q$.
    \end{enumerate}
    Moreover, the sets $W$ and $D$ and the components $C_1,\ldots,C_q$ can be computed in time polynomial in $|V(G)|$.
\end{lemma}
\begin{proof}
    First, observe that $S \subseteq W$: indeed, as $S$ is independent, for every $x \in S$ we have $N(x) \cap S = \emptyset$.
    In particular, every vertex of $D$ is outside $S$, so \eqref{eq:extract-is} applies to it.
    Next, note that for every $x \in W$ we have $|N[x] \cap S| \leq s$: if $x \notin S$, this is by the definition of $W$, and for $x \in S$ we have $N[x] \cap S = \{x\}$.
    For contradiction, suppose that $W$ contains a set $P$ of four vertices inducing a $P_4$. Consequently, $|N[P] \cap S| \leq 4s$. As $|S| \geq \sigma = sq \geq 5s$, where the last inequality follows from $q \geq 2s+3 \geq 5$, there are $s$ vertices in $S$ that are not adjacent to any vertex of $P$. Together with $P$, these vertices induce a $P_4+sP_1$ in $G$, a contradiction.
    This proves the first item.

      \begin{claim}\label{clm:componentswithS}
        Every component of $G[W]$ contains at most $s$ vertices from $S$.
    \end{claim}
    \begin{claimproof}
        Let $C$ be the vertex set of a component of $G[W]$ and suppose for contradiction that $|C \cap S| > s$.
        As $s \geq 1$, we have $|C \cap S| \geq 2$, and $C \cap S$ is an independent set in the connected cograph $G[C]$.
        Hence, by \cref{lem:dominator-cographs}, there is a vertex $v \in C \setminus S$ complete to $C \cap S$.
        But then $v$ is adjacent to more than $s$ vertices of $S$, so $v \notin W$, a contradiction.
    \end{claimproof}

    Combining \cref{clm:componentswithS} with the fact that $|S| \geq \sigma = sq$, we obtain that there are at least $q$ distinct components $C_1,\ldots,C_q$ of $G[W]$ that intersect $S$, i.e., the second item holds.

    For each $i \in [q]$, let $u_i$ be an arbitrary vertex from $C_i \cap S$.
    Pick any $x \in D$. By the definition of $D$, we have $|N(x) \cap S| > s$, which, by \eqref{eq:extract-is}, implies that $|S \setminus N(x)| \leq s$.
    For contradiction, suppose that $x$ is not complete to at least $s+1$ among components $C_1,\ldots,C_q$. By symmetry, suppose that $x$ is not complete to $C_1,\ldots,C_{s+1}$.
    As $x$ has at most $s$ non-neighbors in $S$, there is $i \in [s+1]$ such that $x$ is adjacent to $u_i \in C_i \cap S$. By symmetry, suppose that $xu_{s+1} \in E(G)$.
    Consequently, $x$ is mixed on $C_{s+1}$ and thus, by \cref{lem:mixed}, there is an edge $uv$ contained in $C_{s+1}$ such that $xv \notin E(G)$ and $xu \in E(G)$.

    For $i \in [s]$, pick any $v_i \in C_i \setminus N(x)$; it exists as $x$ is not complete to any of $C_1,\ldots,C_s$.
    Finally, recall that $x$ has at most $s$ non-neighbors in $S$.
    As $q \geq 2s+3$, we have $q-(s+1) > s$, and hence there is some $j \in \{s+2,\ldots,q\}$ such that $x$ is adjacent to $u_j \in C_j \cap S$.

    Summing up, we claim that the vertices $v,u,x,u_j$ together with $v_1,\ldots,v_s$ induce a $P_4+sP_1$ in $G$.
    Indeed, $vu,ux,xu_j \in E(G)$, while $vx \notin E(G)$ and $vu_j,uu_j \notin E(G)$ as $v,u \in C_{s+1}$ and $u_j \in C_j$ are in distinct components of $G[W]$; hence $v,u,x,u_j$ induce a $P_4$.
    Furthermore, for every $i \in [s]$, the vertex $v_i$ is nonadjacent to $x$ by its choice, and it is nonadjacent to $v,u,u_j$ and to every $v_{i'}$ with $i' \neq i$, as all these vertices lie in components of $G[W]$ distinct from $C_i$.
    This is a contradiction, which proves the third item.

    Finally, the sets $W$ and $D$ are computed directly from their definition, and the connected components of $G[W]$, along with the information which of them intersect $S$, are found by a standard graph search; all this takes polynomial time.
\end{proof}

\paragraph{Reduction to subinstance.}
The next lemma is the key step of our approach.
We show that given a $(P_4+sP_1)$-free graph $G$ with the structure provided by \cref{lem:components}, we can reduce solving \lcol{$k$} on $G$ to solving \lcol{$k$} on certain induced subgraphs of $G$.
For simplicity of notation, we use a convention that the list function might be defined on a superset of the vertex set of the graph. In particular, we do not explicitly restrict the lists when discussing instances obtained by deleting some vertices.
For a list function $L$ and a set of colors $K$, we denote by $L|K$ the restriction of $L$ to the colors in $K$, i.e., $L|K(v) = L(v) \cap K$ for every vertex $v$.

\begin{lemma}\label{lem:reduce}
    Let $q$ be as in \eqref{eq:constants} and let $(G,L)$ be an instance of \lcol{$k$}.
    Let $W,D$ be a partition of $V(G)$ and let $C_1,\ldots,C_q$ be distinct connected components of $G[W]$, called \emph{distinguished}, such that every vertex of $D$ is complete to all but at most $s$ of the components $C_1,\ldots,C_q$.

    For every $K \subseteq [k]$, let $\cC_K$ consist of any $sk+1$ distinguished components $C_i$ for which $(G[C_i],L|K)$ is a no-instance of \lcol{$k$}, or all such components if there are fewer than $sk+1$ of them.
    Let $Z$ be the union of the distinguished components that do not appear in $\cC_K$ for any $K \subseteq [k]$.
    Then $Z \neq \emptyset$ and $(G,L)$ is a yes-instance of \lcol{$k$} if and only if $(G-Z,L)$ is a yes-instance of \lcol{$k$}.

\end{lemma}
\begin{proof}
    First, note that as $q > (sk+1)2^k$ by \eqref{eq:constants-prop} and there are at most $2^k$ distinct sets $K \subseteq [k]$, we have $Z \neq \emptyset$.
    Furthermore, if $(G,L)$ is a yes-instance of \lcol{$k$}, then so is $(G-Z,L)$, as $G-Z$ is an induced subgraph of $G$ and the lists are not changed.

    It remains to show the converse. Let $\phi$ be a proper coloring of $G-Z$ respecting $L$.
    Note that $Z$ is a union of some components of $G[W]$, so $Z \subseteq W$ and thus $D$ is contained in $G-Z$.
    Let $A$ be the set of colors used on $D$, i.e., $A = \phi(D)$, and let $K = [k] \setminus A$.

    We say that a distinguished component $C_i$ is \emph{surviving} if it was not included in $Z$. Otherwise, a distinguished component is \emph{deleted}.

    \begin{claim}\label{clm:surviving-components}
        At most $sk$ distinguished surviving components contain a vertex colored by $\phi$ with a color from $A$.
    \end{claim}
    \begin{claimproof}
        Fix some color $a \in A$ and let $x_a \in D$ be a vertex with $\phi(x_a) = a$.
        Consider a distinguished surviving component $C$ such that $a \in \phi(C)$.
        Clearly, this means that $x_a$ is not complete to $C$.
        As $x_a$ is complete to all but at most $s$ of the distinguished components, at most $s$ distinguished surviving components may use color $a$.
        Summing over all $a \in A$ and recalling that $A \subseteq [k]$, and thus $|A| \leq k$, we obtain the claim.
    \end{claimproof}

    Now, let us discuss deleted components.

    \begin{claim}\label{clm:deleted-components}
        Let $C$ be a deleted component.
        Then $(G[C],L|K)$ is a yes-instance of \lcol{$k$}.
    \end{claim}
    \begin{claimproof}
        Suppose for contradiction that $(G[C],L|K)$ is a no-instance of \lcol{$k$}.
        Since $C$ is deleted, it was not included in $\cC_K$, which means that $\cC_K$ already contained $sk+1$ other components.
        In particular, there are $sk+1$ surviving distinguished components $C_i$ such that $(G[C_i],L|K)$ is a no-instance of \lcol{$k$}.

        Consider such a component $C_i$; as it is surviving, it is contained in $G-Z$ and thus $\phi$ colors all its vertices.
        The restriction of $\phi$ to $C_i$ is a proper coloring of $G[C_i]$ respecting $L$.
        If it used only colors from $K$, it would also respect $L|K$, and thus $(G[C_i],L|K)$ would be a yes-instance.
        Consequently, on each such $C_i$, the coloring $\phi$ must use some color that is not in $K$, i.e., some color from $A$.
        This contradicts \cref{clm:surviving-components}, as there are at most $sk$ surviving distinguished components that use colors from $A$.
    \end{claimproof}

    For each deleted component $C$, let $\phi_C$ be a proper coloring of $G[C]$ respecting the lists $L|K$; it exists by \cref{clm:deleted-components}.
    We extend $\phi$ to a coloring of $G$ by coloring each deleted component $C$ with $\phi_C$.

    Clearly, the resulting coloring respects the lists.
    Furthermore, it is proper.
    Indeed, consider an edge with an endpoint in a deleted component $C$.
    Its other endpoint is not in $W \setminus C$, as $C$ is a connected component of $G[W]$, so it is either in $C$ or in $D$.
    In the first case, the two endpoints receive distinct colors as $\phi_C$ is proper.
    In the second case, the endpoint in $C$ receives a color from $K$, while no vertex from $D$ receives a color from $K$ in the coloring $\phi$.
\end{proof}

The following algorithmic corollary summarizes the results obtained so far.

\begin{corollary}\label{cor:reduce}
    Let $(G,L)$ be an instance of \lcol{$k$} such that $G$ is $(P_4+sP_1)$-free and has $n$ vertices, and let $M$ be as in \eqref{eq:constants}.
    In time $2^{\Oh(k)} \cdot n^{\Oh(1)}$ we can return one of the following outputs:
    \begin{enumerate}
        \item a conclusion that a largest independent set in $G$ has size smaller than $M$, or
        \item a clique of size $k+1$ in $G$, or
        \item a nonempty set $Z \subseteq V(G)$ such that, if $G-Z$ admits a proper coloring respecting $L$, then so does $G$.
    \end{enumerate}
\end{corollary}
\begin{proof}
    We start by finding a largest independent set $X$ in $G$ using \cref{thm:mis-p4sp1}.
    If $|X| < M$, then we return the first output.
    Otherwise, we apply \cref{lem:extract-is} to $G$ and $X$, with $p := \sigma$;
    its assumptions are satisfied, as $|X| \geq M = h_s(2k,\sigma)$ and $\sigma > s$ by \eqref{eq:constants-prop}.
    If it returns a clique of size $k+1$, we return the second output.
    Otherwise, we obtain an independent set $S \subseteq X$ of size $\sigma$ that satisfies \eqref{eq:extract-is} for every $v \in V(G) \setminus S$.

    Now we apply \cref{lem:components} to $G$ and $S$; this is legitimate, as $q \geq 2s+3$ by \eqref{eq:constants-prop} and $|S| = \sigma = sq$.
    We obtain a partition $W,D$ of $V(G)$ and distinguished components $C_1,\ldots,C_q$ of $G[W]$.

    For every $K \subseteq [k]$ and every $C_i$ for $i \in [q]$,
    we call the algorithm from \cref{thm:lcol-cographs} on $(G[C_i],L|K)$; recall that $G[C_i]$ is a cograph, as $C_i \subseteq W$ and $G[W]$ is $P_4$-free.
    This way we build families $\cC_K$ for every $K \subseteq [k]$ as in \cref{lem:reduce}.
    Finally, we return the set $Z$ as in \cref{lem:reduce}; it is nonempty and has the desired property.

    Let us discuss the running time.
    Finding $X$ takes time $n^{\Oh(1)}$ by \cref{thm:mis-p4sp1}, and applying \cref{lem:extract-is} takes time $n^{\Oh(1)}$;
    in both cases the exponent in the polynomial depends on $s$.
    Computing $W$, $D$, and the distinguished components takes time $n^{\Oh(1)}$ with an absolute constant in the exponent.
    Finally, we perform $2^k \cdot q$ calls to the algorithm from \cref{thm:lcol-cographs}, each running in time $2^{\Oh(k)} \cdot n^{\Oh(1)}$.
    As $q = (sk+1)2^k+1 = 2^{\Oh(k)}$ by \eqref{eq:constants}, the number of calls is $2^{\Oh(k)}$, so this is $2^{\Oh(k)} \cdot n^{\Oh(1)}$ in total.
    Summing up, the total running time is $2^{\Oh(k)} \cdot n^{\Oh(1)}$.
    This completes the proof of the corollary.
\end{proof}

\paragraph{Wrapping up the proof.}
Finally, we are ready to prove \cref{thm:p4sp1algo}.

\thmpfour*

\begin{proof}
    If $s=0$, then the statement follows from \cref{thm:lcol-cographs}, so we may assume that $s \geq 1$.
    Let $(G,L)$ be an instance of \lcol{$k$} such that $G$ is $(P_4+sP_1)$-free and has $n$ vertices.
    Recall that $M = 2^{\Oh(k^2)}$ by \eqref{eq:M-bound}, and thus also $kM = 2^{\Oh(k^2)}$.
    If $n \leq kM$, then we solve the instance by calling the \XP algorithm from \cref{thm:lcol-p5sp1}. This results in the running time $(2^{\Oh(k^2)})^{k^3} =  2^{\Oh(k^5)}$.
    Thus, we may assume that $n > kM$.

    We apply \cref{cor:reduce} to $(G,L)$.
    If the first or the second output is returned, then we reject the instance.
    In case of the first output this is justified, as any $k$-colorable graph with a largest independent set of size smaller than $M$ has fewer than $kM$ vertices. Thus, $(G,L)$ is a no-instance.
    Similarly, the second output, i.e., a clique of size $k+1$, is a certificate that $(G,L)$ is a no-instance.

    Finally, if the third output is returned, we obtain a nonempty set $Z \subseteq V(G)$.
    We delete it from the graph and start over with the reduced instance $(G-Z,L)$; note that $G-Z$ is still $(P_4+sP_1)$-free, as the class of $(P_4+sP_1)$-free graphs is hereditary.
    This is justified, as $(G,L)$ and $(G-Z,L)$ are equivalent:
    if $(G-Z,L)$ is a yes-instance, then so is $(G,L)$ by \cref{cor:reduce}, and conversely, the restriction of any proper coloring of $G$ respecting $L$ to $V(G) \setminus Z$ is a proper coloring of $G-Z$ respecting $L$.
    As at each step we delete at least one vertex, this procedure is repeated at most $n$ times.

    Consequently, the algorithm consists of at most $n$ iterations, each taking time $2^{\Oh(k)} \cdot n^{\Oh(1)}$, and the application of \cref{thm:lcol-p5sp1} on the remaining graph of size at most $kM$.
    Summing up, the total running time is $2^{\Oh(k^5)} \cdot n^{\Oh(1)}$.
    This completes the proof.
\end{proof}

\newpage
\section{\col{$k$} in $2P_2$-free graphs}
In this section we prove \cref{thm:2p2hard}.
Let us first prove an analogous hardness result for \lcol{$k$}. The constructed instances will satisfy certain additional properties, which will be useful in the proof of \cref{thm:2p2hard}.

\begin{theorem}\label{thm:list2p2hard}
    \lcol{$k$} on $2P_2$-free graphs is \Wone-hard when parameterized by $k$, even when every list has size at most $4$.
    Furthermore, the problem on $n$-vertex instances cannot be solved in time $f(k) \cdot n^{o(k)}$, for any computable function $f$, unless the ETH fails.

    \noindent These lower bounds hold even if every edge $xy$ of the instance $(G,L)$ satisfies the following two properties:
    \begin{itemize}
        \item $|L(x) \cap L(y)|\leq 1$,
        \item if $c \in L(x) \cap L(y)$, and $w$ is a vertex non-adjacent to both $x$ and $y$, then $c \in L(w)$.
    \end{itemize}
\end{theorem}

\begin{proof}
We reduce from \Clique.
Let $(G',r)$ be an instance, and let $V^1,\ldots,V^r$ be the partition of $V(G')$, where every part has size $\ell$.
For every $i \in [r]$, fix an ordering of $V^i$ and write its vertices as $v^i_1,\ldots,v^i_\ell$.

We will create an instance $(G,L)$ of \lcol{$k$}, where $k = 4r$,
such that $G$ is $2P_2$-free and admits a proper coloring respecting $L$ if and only if $(G',r)$ is a yes-instance of \Clique.
The colors used in the lists are $\bigcup_{i \in [r]} \{c_0^i,c_1^i,c_2^i,c_3^i\}$.

\paragraph{Selector gadget.}
For each $i \in [r]$, we introduce a \emph{selector gadget} $S^i$.
Fix $i \in [r]$.
The vertex set of the gadget is partitioned into four independent sets $X^i_0,X^i_1,X^i_2,X^i_3$ called \emph{layers}.
For $a \in \{0,1,2,3\}$, we have $X^i_a = \{x^i_{a,p} \mid p \in [\ell+1]\}$.
Arithmetic operations on the index $a$ are performed modulo $4$.
We also drop the upper index if it is clear from the context.

The edges between the sets are as follows. First, we add all edges between $X_0$ and $X_2$, and between $X_1$ and $X_3$. Second, for each $a \in \{0,1,2,3\}$ and each $p,q \in [\ell+1]$, we add an edge $x_{a,p}x_{a+1,q}$ if and only if $p \geq q$.

Finally, the lists of vertices of the gadget are set as follows:
\[
L(x_{a,p}) = \begin{cases}
\{c_{a+1}^i\} & \text{ if } p=1,\\
\{c_{a}^i,c_{a+1}^i\} & \text{ if } p \in \{2,\ldots,\ell\},\\
\{c_a^i\} & \text{ if } p=\ell+1.\\
\end{cases}
\]

We say that a coloring $\phi$ of the selector gadget \emph{selects} a vertex $v^i_t \in V^i$
if, for all $a \in \{0,1,2,3\}$ and $p \in [\ell+1]$, we have
\[
        \phi(x_{a,p}) =
        \begin{cases}
            c_{a+1}^i & \text{ if } p \leq t,\\
            c_a^i & \text{ if } p > t.
        \end{cases}
\]
Note that such a coloring is unique.

\begin{claim}\label{clm:from-selection-to-coloring}
    For every $t \in [\ell]$, the coloring of $S^i$ that selects $v^i_t$ is proper and respects the lists.
\end{claim}
\begin{claimproof}
    The prescribed colors respect the lists. Edges between non-consecutive layers join vertices with disjoint lists.
    Two vertices $x_{a,p}$ and $x_{a+1,q}$ from consecutive layers receive the same color only when
    $p \leq t < q$. In this case $p<q$, and hence these vertices are not adjacent.
\end{claimproof}

\begin{claim}\label{clm:from-coloring-to-selection}
    Every proper coloring $\phi$ of $S^i$ respecting the lists selects a vertex $v^i_t \in V^i$.
\end{claim}
\begin{claimproof}
    For each $a \in \{0,1,2,3\}$, let $t_a$ be the largest index $p$
    such that $\phi(x_{a,p})=c^i_{a+1}$.
    This is well-defined and belongs to $[\ell]$,
    since $\phi(x_{a,1})=c^i_{a+1}$ and $\phi(x_{a,\ell+1})=c^i_a$.

    We claim that $t_a \leq t_{a+1}$ for every $a$ (where indices are computed modulo 4).
    Indeed, for every $q \leq t_a$, the vertices $x_{a,t_a}$ and $x_{a+1,q}$ are adjacent.
    Since $\phi(x_{a,t_a})=c^i_{a+1}$,
    properness of $\phi$ and the list of $x_{a+1,q}$ imply that $\phi(x_{a+1,q})=c^i_{a+2}$.
    In particular, $t_{a+1} \geq t_a$.

    As the indices $a$ are taken modulo $4$, these inequalities give
    \[
        t_0 \leq t_1 \leq t_2 \leq t_3 \leq t_0.
    \]
    Thus, all four indices are equal; denote their common value by $t$.
    Applying the preceding argument with $a-1$ in place of $a$ shows that
    $\phi(x_{a,p})=c^i_{a+1}$ for every $p \leq t$. For $p>t$, the
    maximality of $t_a=t$ and the list of $x_{a,p}$ imply that
    $\phi(x_{a,p})=c^i_a$.
    Thus, $\phi$ indeed selects $v^i_t$.
\end{claimproof}

By \cref{clm:from-selection-to-coloring,clm:from-coloring-to-selection}, proper list colorings of $S^i$ are in one-to-one correspondence with vertices of $V^i$.

\paragraph{Nonedge gadget.}
    Let $1 \leq i<j \leq r$ and let $v^i_pv^j_q \notin E(G')$.
    The \emph{nonedge gadget} introduced for this nonedge is a single vertex $y^{i,j}_{p,q}$ with list $\{c_0^i,c_1^i,c_0^j,c_1^j\}$.
    We make $y^{i,j}_{p,q}$ adjacent to the following vertices from $S^i$ and $S^j$ (see also \cref{fig:nonedgegadget}):
    \begin{itemize}
        \item the vertices from $X^i_2$ and $X^j_2$,
        \item the vertices $x^i_{1,p'}$ for $p' \leq p+1$ and $x^i_{3,p'}$ for $p' \geq p$,
        \item the vertices $x^j_{1,q'}$ for $q' \leq q+1$ and $x^j_{3,q'}$ for $q' \geq q$.
    \end{itemize}
    By $Y^{i,j}$ we denote the set of all nonedge-gadget vertices introduced for nonedges between $V^i$ and $V^j$,
    and by $Y^i$ the set of all nonedge-gadget vertices introduced for nonedges with one endpoint in $V^i$.

\begin{claim}\label{claim:nonedge-gadget}
    Let $i,j \in [r]$, where $i<j$, and let $v^i_pv^j_q \notin E(G')$.
    Let $y^{i,j}_{p,q}$ be the corresponding nonedge-gadget vertex.
    Let $\phi$ be a proper coloring of $S^i$ and $S^j$ that respects the lists,
    and suppose that it selects $v^i_s \in V^i$ and $v^j_t \in V^j$.
    Then, $\phi$ can be extended to $y^{i,j}_{p,q}$ respecting its list if and only if $(s,t) \neq (p,q)$.
\end{claim}
\begin{claimproof}
    See \cref{fig:nonedgegadget} for an illustration of the argument.
    Since $\phi$ selects $v^i_s$, the color $c_1^i$ appears on a neighbor of
    $y^{i,j}_{p,q}$ in $X^i_1$ if and only if $s \leq p$: indeed, such a neighbor
    can be chosen as $x^i_{1,s+1}$ when $s \leq p$, whereas all neighbors of
    $y^{i,j}_{p,q}$ in $X^i_1$ have index at most $p+1 \leq s$ when $s>p$.
    Similarly, the color $c_0^i$ appears on a neighbor of $y^{i,j}_{p,q}$ in
    $X^i_3$ if and only if $s \geq p$.
    Hence, both colors $c_0^i$ and $c_1^i$ are unavailable for $y^{i,j}_{p,q}$
    if and only if $s=p$.

    The symmetric argument for the selector gadget for $j$ shows that both
    colors $c_0^j$ and $c_1^j$ are unavailable if and only if $t=q$.
    The neighbors in $X^i_2 \cup X^j_2$ receive only colors $c_2^i,c_3^i,c_2^j,c_3^j$,
    none of which belongs to $L(y^{i,j}_{p,q})$.
    Therefore, no color in $L(y^{i,j}_{p,q})$ is available precisely when
    $s=p$ and $t=q$.
\end{claimproof}

\begin{figure}
    \centering
    \includegraphics{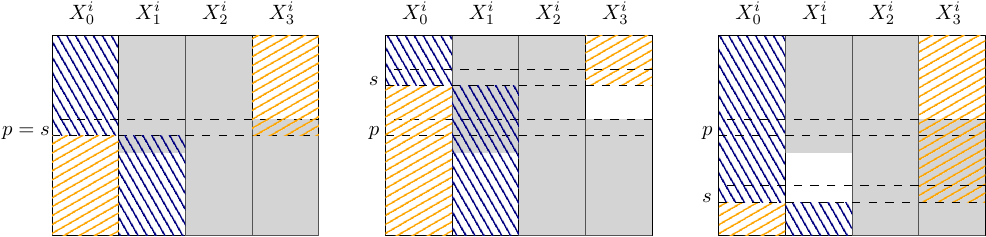}
    \caption{Illustration for \cref{claim:nonedge-gadget}.
        Columns indicate the four layers of the selector gadget for $i$.
        The shaded area depicts the neighborhood of $y^{i,j}_{p,q}$ in this selector gadget.
        Suppose that $\phi$ selects $v^i_s$ --- we distinguish three cases: $s=p$ (left), $s < p$ (middle), or $s > p$ (right).
        Blue (falling pattern) indicates the vertices colored $c^i_1$, while orange (rising pattern) indicates the vertices colored $c^i_0$.
        Only when $s=p$ are both colors $c^i_0$ and $c^i_1$ unavailable for $y^{i,j}_{p,q}$.}
    \label{fig:nonedgegadget}
\end{figure}

\paragraph{Completing the construction.}
    Finally, we add all edges between vertices of $S^i$ and $S^j$ for $i \neq j$,
    and, for all $1 \leq i<j \leq r$ and $h \in [r] \setminus \{i,j\}$, all edges between $Y^{i,j}$ and $S^h$.
    Note that in this step we add edges between vertices with disjoint lists, so these edges do not influence the coloring of $G$.
    This completes the definition of $(G,L)$.
    Clearly, the instance can be constructed in polynomial time.

    \paragraph{Equivalence of instances.}
    Now, let us argue that $(G,L)$ is a yes-instance of \lcol{$k$} if and only if $(G',r)$ is a yes-instance of \Clique.

\begin{claim}
    Let $\phi$ be a proper coloring of $G$ that respects the lists $L$,
    and, for $i \in [r]$,
    let $v^i \in V^i$ be the vertex of $V^i$ selected by $\phi$ restricted to $S^i$.
    Then, the set $\{v^1,\ldots,v^r\}$ is a clique in $G'$.
\end{claim}
\begin{claimproof}
    For contradiction, suppose that $v^iv^j \notin E(G')$, where $i<j$.
    Write $v^i=v^i_p$ and $v^j=v^j_q$. The nonedge gadget $y^{i,j}_{p,q}$ exists,
    and \cref{claim:nonedge-gadget} shows that the restriction of $\phi$ to its
    neighbors cannot be extended to $y^{i,j}_{p,q}$, a contradiction.
\end{claimproof}

\begin{claim}
    For every clique $C$ of $G'$ that intersects each set $V^i$ in exactly one vertex,
    there exists a proper coloring $\phi$ of $G$ that respects the lists $L$.
\end{claim}
\begin{claimproof}
    Color every selector gadget $S^i$ according to the vertex of $C \cap V^i$ that it selects.
    Consider a nonedge gadget $y^{i,j}_{p,q}$.
    Since $C$ is a clique, the vertices of $C$ in $V^i$ and $V^j$ are not $v^i_p$ and $v^j_q$, respectively.
    Thus, \cref{claim:nonedge-gadget} gives an available color for $y^{i,j}_{p,q}$.
    As nonedge gadgets form an independent set, we may choose these colors independently.
    As all edges added in the finishing step join vertices with disjoint lists,
    the obtained coloring is proper and respects the lists.
\end{claimproof}

\paragraph{$2P_2$-freeness.}
Now, let us verify that $G$ is $2P_2$-free.
We will do it in several steps.

\begin{claim}\label{clm:si-2p2}
    For each $i \in [r]$, the selector gadget $S^i$ is $2P_2$-free.
\end{claim}
\begin{claimproof}
    For contradiction, suppose that there is $i \in [r]$ for which $S^i$ contains an induced copy of $2P_2$ with edges $vv'$ and $uu'$.
    To simplify the notation, we drop the upper index $i$.

    Suppose first that one of the two edges joins two consecutive layers $X_a$ and $X_{a+1}$.
    By symmetry, assume that $a=0$ and $v \in X_0$ and $v' \in X_1$.
    Note that this means that $u,u' \notin X_2 \cup X_3$ as $X_0$ is complete to $X_2$ and $X_1$ is complete to $X_3$.
    Consequently, both edges $vv'$ and $uu'$ join layers $X_0$ and $X_1$.
    However, the graph between these two layers is a chain graph, and hence it is $2P_2$-free.

    Thus, both edges join non-consecutive layers.
    Note that they cannot join the same pair of layers, as the graph between $X_0$ and $X_2$,
    as well as between $X_1$ and $X_3$, is complete bipartite and hence $2P_2$-free.

    So, by symmetry, we may assume that $v \in X_0$, $v' \in X_2$, $u \in X_1$, and $u' \in X_3$.
    Let $p,p',q,q'$ be such that $v=x_{0,p}$, $v'=x_{2,p'}$, $u=x_{1,q}$, and $u'=x_{3,q'}$.
    As $v$ is not adjacent to $u$, we have $p < q$.
    Repeating the same argument yields the following sequence of inequalities:
    \[
        p < q < p' < q' < p,
    \]
    a contradiction.
\end{claimproof}

Define $H^i=G[V(S^i) \cup Y^i]$.

\begin{claim}\label{clm:hi-2p2}
    For each $i \in [r]$, the graph $H^i$ is $2P_2$-free.
\end{claim}
\begin{claimproof}
    Fix some $i \in [r]$. When referring to vertices of $S^i$, we drop the upper index $i$ for simplicity of notation.

    Every vertex $y \in Y^i$ has an associated index $t=t(y) \in [\ell]$,
    namely the index of its endpoint in $V^i$, such that
    \begin{equation}
        N^i(y) := N(y) \cap V(S^i) = X_2 \cup \{x_{1,p} \mid p \leq t+1\}
        \cup \{x_{3,p} \mid p \geq t\}.\label{eq:defNi}
    \end{equation}

    For contradiction, suppose that $H^i$ contains an induced copy of $2P_2$.
    Note that at least one of the four vertices of this copy must belong to $Y^i$, as $S^i$ is $2P_2$-free by \cref{clm:si-2p2}.

    Consider two vertices $y,y'$ in $Y^i$, with associated indices $t(y) \leq t(y')$.
    By \eqref{eq:defNi}, every vertex in $N^i(y) \setminus N^i(y')$ belongs to $X_3$,
    whereas every vertex in $N^i(y') \setminus N^i(y)$ belongs to $X_1$.
    Since $X_1$ is complete to $X_3$, no two edges incident to, respectively, $y$ and $y'$,
    may form an induced $2P_2$.

    Since $Y^i$ is independent we are left with the case that the considered copy of $2P_2$
    is formed by an edge $yu$ where $y \in Y^i$ and $u \in V(S^i)$,
    together with an edge of $S^i$.
    Let $t=t(y)$. By \eqref{eq:defNi}, every vertex outside $N^i(y)$ belongs to $X_0$, to $\{x_{1,p} \mid p\geq t+2\}$,
    or to $\{x_{3,p} \mid p\leq t-1\}$. Consequently, an edge with both endpoints
    outside $N^i(y)$ has one of the following forms:
    \begin{description}
        \item[Case 1.] $x_{0,a}x_{1,b}$ for $a\geq b\geq t+2$, or
        \item[Case 2.] $x_{0,a}x_{3,b}$ for $b\geq a$ and $b\leq t-1$, or
        \item[Case 3.] $x_{1,a}x_{3,b}$ for $a\geq t+2$ and $b\leq t-1$.
    \end{description}

    Let us discuss possible locations of $u$. Note that
    \[
    u \in \{x_{1,p} \mid p\leq t+1\} \cup X_2 \cup \{x_{3,p} \mid p \geq t\}.
    \]

    \begin{description}
        \item[Case 1.]
        If $u = x_{1,p}$, then $p \leq t+1 < b \leq a$, and thus $u$ is adjacent to $x_{0,a}$.
        If $u \in X_2$, then it is adjacent to $x_{0,a}$.
        If $u \in X_3$, then it is adjacent to $x_{1,b}$.

        \item[Case 2.]
        If $u \in X_1$, then it is adjacent to $x_{3,b}$.
        If $u \in X_2$, then it is adjacent to $x_{0,a}$.
        Finally, if $u = x_{3,p}$, then $p\geq t>b\geq a$ and thus $u$ is adjacent to $x_{0,a}$.

        \item[Case 3.]
        If $u \in X_1$, then it is adjacent to $x_{3,b}$.
        If $u =x_{2,p}$, then either $p \leq a$ and $u$ is adjacent to $x_{1,a}$,
        or $p > a \geq b+2 > b$ and $u$ is adjacent to $x_{3,b}$.
        Finally, if $u \in X_3$, then it is adjacent to $x_{1,a}$.
    \end{description}
    Summing up, in all cases $u$ is adjacent to an endpoint of the edge outside $N^i(y)$.
    Thus, we conclude that $H^i$ is $2P_2$-free.
\end{claimproof}

Finally, we are ready to prove that $G$ is $2P_2$-free.
\begin{claim}
    $G$ is $2P_2$-free.
\end{claim}
\begin{claimproof}
    For contradiction, suppose now that $G$ contains an induced $2P_2$; denote its edges by $vv'$ and $uu'$.
    Since nonedge gadgets form an independent set, each of these two edges has at least one endpoint in a selector gadget. By symmetry, let $u$ and $v$ be such endpoints. Note that they both belong to the selector gadget, say, $S^i$, since otherwise they would be adjacent.
    Similarly, none of $u',v'$ belongs to $S^j$ for $j \neq i$.

    Suppose now that one of $u',v'$, say, $v'$, is a nonedge-gadget vertex that does not belong to $Y^i$.
    Then $v'$ is complete to $S^i$ by the finishing step, and hence it is adjacent to $u$, a contradiction.

    Consequently, all four vertices $u,u',v,v'$ belong to the vertex set of $H^i$.
    This means that $H^i$ contains an induced $2P_2$, a contradiction with \cref{clm:hi-2p2}.
\end{claimproof}

\paragraph{Additional properties of edges.}
Finally, let us show that $(G,L)$ satisfies the additional properties in the theorem statement.

\begin{claim}
    Every edge $xy$ of the instance $(G,L)$ satisfies the following two properties:
    \begin{itemize}
        \item $|L(x) \cap L(y)|\leq 1$,
        \item if $c \in L(x) \cap L(y)$, and $w$ is a vertex non-adjacent to both $x$ and $y$, then $c \in L(w)$.
    \end{itemize}
\end{claim}
\begin{claimproof}
    We first identify the edges whose endvertices have a common color in their
    lists. Edges added in the final step join vertices with disjoint lists.
    The same holds for the edges between $X^i_0$ and $X^i_2$, and between $X^i_1$
    and $X^i_3$, for every $i \in [r]$.
    For a consecutive-layer edge
    \[
        x^i_{a,p}x^i_{a+1,q},
    \]
    the only possible common color is $c^i_{a+1}$.
    Finally, let $y=y^{i,j}_{p,q}$ be a nonedge gadget. The only edges incident with $y$ whose
    endvertices have a common color are, up to symmetry between $i$ and $j$,
    \[
        yx^i_{1,s} \text{ with common color } c^i_1,
        \qquad
        yx^i_{3,s} \text{ with common color } c^i_0.
        \tag{$\star$}
    \]
    Thus, the intersection of the lists at the ends of every edge indeed has size at
    most one.

    It remains to prove the second property.
    Fix an edge $xy$ with a common color $c$.
    We will show that every vertex $w$ with $c \notin L(w)$ is adjacent to $x$ or to $y$.

    First, suppose that $xy=x^i_{a,p}x^i_{a+1,q}$ is a consecutive-layer edge,
    so $c=c^i_{a+1}$. Within $S^i$, every vertex in $X^i_{a+2}$ is adjacent to
    $x$, and every vertex in $X^i_{a+3}$ is adjacent to $y$.
    The only vertex of $X^i_a$ whose list does not contain $c$ is $x^i_{a,\ell+1}$, which is adjacent
    to $y$, and the only vertex of $X^i_{a+1}$ whose list does not contain $c$
    is $x^i_{a+1,1}$, which is adjacent to $x$.
    Hence the assertion holds for every vertex of $S^i$.

    Consider now a vertex $w$ outside $S^i$.
    If $w$ belongs to another selector gadget, or is a nonedge-gadget vertex not in $Y^i$,
    then it is complete to $S^i$ by the final step of the construction.
    So suppose that $w \in Y^i$.
    If $a\in\{0,3\}$, then $c\in\{c^i_0,c^i_1\}\subseteq L(w)$;
    if $a\in\{1,2\}$, then one endpoint of $xy$ belongs to $X^i_2$ and is adjacent to $w$.
    This proves the assertion in the consecutive-layer case.

    It remains to consider an edge as in $(\star)$.
    Let $x$ belong to $S^i$ and suppose first that $w$ also belongs to $S^i$.
    If $x = x^i_{1,s}$, then the common color is $c^i_1$.
    The vertices of $S^i$ whose lists do not contain $c^i_1$ are $x^i_{0,\ell+1}$, $x^i_{1,1}$,
    all vertices of $X^i_2$, and all vertices of $X^i_3$.
    The first is adjacent to $x^i_{1,s}$,
    the second and the vertices of $X^i_2$ are adjacent to the nonedge gadget,
    and every vertex of $X^i_3$ is adjacent to $x^i_{1,s}$.
    If $x= x^i_{3,s}$, then the common color is $c^i_0$.
    In this case, the vertices of $S^i$ whose lists do not contain $c^i_0$ are $x^i_{0,1}$,
    all vertices of $X^i_1$, all vertices of $X^i_2$, and $x^i_{3,\ell+1}$.
    The first and all vertices of $X^i_1$ are adjacent to $x^i_{3,s}$, while the remaining
    vertices are adjacent to the nonedge gadget.

    Finally, suppose that $w$ does not belong to $S^i$.
    If $w$ is in a selector gadget $S^j$ for $j \neq i$, then it is adjacent to $x$.
    If $w \in Y^i$, then it has both $c^i_0$ and $c^i_1$ in its list,
    and if $w$ is a nonedge-gadget vertex not in $Y^i$, then it is complete to $S^i$ and thus adjacent to $x$.

    Thus, every vertex whose list omits the common color is adjacent to one of the endpoints of the edge, as required. This completes the proof of the claim.
\end{claimproof}

Summing up, we showed that $(G,L)$ is a yes-instance of \lcol{$k$} if and only if
$(G',r)$ is a yes-instance of \Clique.
All lists have size at most $4$ by construction, and the final two claims establish the $2P_2$-freeness and the
additional list properties stated in the theorem.

Recall that
\[
    |V(G)|\leq 4r(\ell+1)+\binom{r}{2}\ell^2=\Oh(|V(G')|^2)
\]
and $k=4r$. Thus, an algorithm
running in time $f(k) \cdot n^{o(k)}$ on $n$-vertex instances would solve \Clique in time
$f(4r) \cdot |V(G')|^{o(r)}$, contradicting \cref{thm:msi}. This completes the proof.
\end{proof}

Now we are ready to prove our main result.

\thmtwoptwo*
\begin{proof}
    We reduce from the hard instances of \lcol{$k$} given by \cref{thm:list2p2hard}.
    Let $(G',L)$ be such an instance. By relabeling colors, we may assume that
    every list is a subset of $[k]$. Thus, for every edge $xy$ of $G'$, we have
    \begin{itemize}
        \item $|L(x) \cap L(y)|\leq 1$, and
        \item if $i \in L(x) \cap L(y)$ and $w$ is non-adjacent to both $x$ and $y$, then $i \in L(w)$.
    \end{itemize}

    We apply the standard reduction from \lcol{$k$} to \col{$k$}.
    Let $G$ be obtained from $G'$ by adding a clique
    $K=\{q_1,\ldots,q_k\}$ and, for every $v \in V(G')$ and $i \in [k]$,
    adding the edge $vq_i$ if and only if $i \notin L(v)$.
    It is straightforward to verify that $G$ is $k$-colorable if and only if $(G',L)$ is a yes-instance
    of \lcol{$k$}.

    It remains to show that $G$ is $2P_2$-free. Suppose, for contradiction,
    that $G$ contains an induced $2P_2$. At most two vertices of this induced
    subgraph belong to $K$, since $K$ is a clique. If none belongs to $K$, then
    $G'$ contains an induced $2P_2$, a contradiction.

    Suppose next that two vertices of the induced $2P_2$ belong to $K$, say $q_i$ and $q_j$.
    They form one of the edges of $2P_2$, and the other edge is some edge $xy$ of $G'$.
    Since $q_i$ and $q_j$ are non-adjacent to both $x$ and $y$,
    we have $i,j \in L(x) \cap L(y)$, contradicting $|L(x) \cap L(y)|\leq 1$.

    Finally, suppose that exactly one vertex of the induced $2P_2$ belongs to
    $K$, say $q_i$.
    Its edge has the form $q_iw$ with $w \in V(G')$, and the other edge has the form $xy$ with $x,y \in V(G')$.
    The absence of the edges $q_ix$ and $q_iy$ gives $i \in L(x) \cap L(y)$,
    and $w$ is non-adjacent to both $x$ and $y$. The second property above
    yields $i \in L(w)$, contradicting the edge $q_iw$.

    The construction can be performed in polynomial time and it preserves the parameter.
    Thus, the \Wone-hardness and the ETH lower bound from \cref{thm:list2p2hard}
    transfer to \col{$k$}.
\end{proof}

\newpage
\section{Hardness of 6-Coloring in $(P_4+P_2)$-free graphs}
In this section we prove \cref{thm:p4p2hard}.

\thmsixcoloring*

\setcounter{section}{1}          
\setcounter{theorem}{3}          
\setcounter{claim}{0}  

\begin{proof}
    Note that it is sufficient to prove the result for $k=6$.
    Indeed, adding a universal vertex to a $(P_4+P_2)$-free graph $G$ results in a $(P_4+P_2)$-free graph, as no vertex of $P_4+P_2$ is adjacent to the remaining five ones.
    As the obtained graph is $(k+1)$-colorable if and only if $G$ is $k$-colorable, and iterating this operation would show hardness for any $k \geq 6$.

    We reduce from \lcol{$3$} in bipartite graphs, which is known to be \NP-hard even if each list has size one or three~\cite{Kratochvil1993}. Such a variant is usually called $3$-\textsc{Precoloring Extension}, and we can think of the vertices with lists of size one as already precolored with a given color.

    Let $G'$ be a bipartite graph with some vertices already colored with colors from $[3]$.
    Let $X',Y'$ be the two parts of the bipartition of $G'$.
    We perform some simplifying steps on $G'$.
    First, as long as $X'$ (resp., $Y'$) contains two vertices precolored with the same color, we merge these vertices into one.
    Second, if $X'$ (resp., $Y'$) contains no vertex precolored with color $i \in [3]$, we introduce a new vertex $x_i$ (resp., $y_i$),
    precolored with color $i$, and add it to $X'$ (resp., $Y'$).
    These two steps preserve bipartiteness and do not affect the existence of a proper coloring extending the given precoloring.
    Keep calling the resulting graph $G'$ and its parts $X',Y'$.

    Thus, we can assume that there are vertices $x_1,x_2,x_3 \in X'$ and $y_1,y_2,y_3 \in Y'$ precolored with colors $1,2,3$, respectively.
    Furthermore, no other vertex is precolored.
    Define
    \begin{align*}
        Q & =  \bigcup_{i \in [3]} \{x_i,y_i\}, \\
        X & = X' \setminus \{x_1,x_2,x_3\}, \\
        Y & = Y' \setminus \{y_1,y_2,y_3\}.
    \end{align*}

    We will build a graph $G$ that has a proper 6-coloring if and only if $G'$ has a proper 3-coloring extending the given precoloring.
    We start by introducing to $G$ all vertices of $G'$.
    The vertices in $Q$ form a clique, and we make $X'$ and $Y'$ complete to each other.
    Note that no further edges inside $X'$ or inside $Y'$ are introduced, i.e., both $X$ and $Y$ are independent sets,
    $X$ is anticomplete to $\{x_1,x_2,x_3\}$, and $Y$ is anticomplete to $\{y_1,y_2,y_3\}$.

    Next, for each $i \in [3]$, we introduce an independent set $Z_i = \{ z_{i,e} \mid e \in E(G')\}$, and we write $Z = Z_1 \cup Z_2 \cup Z_3$.
    For distinct $i,j \in [3]$, the sets $Z_i$ and $Z_j$ are complete to each other.
    Finally, for each $i \in [3]$ and each edge $e = xy \in E(G')$ where $x \in X'$ and $y \in Y'$,
    we add edges from $z_{i,e}$ to $\{x,y\} \cup (Q \setminus \{x_i,y_i\})$.
    This completes the construction of $G$; it can clearly be done in polynomial time.

    For convenience, the set of six colors used to color $G$ will be denoted by $\{1,2,3,1',2',3'\}$, and we define $\pi(i)=\pi(i')=i$ for $i \in [3]$.
    The intuition is that, as $Q$ is a clique of size 6, in every 6-coloring of $G$ we may assume that $x_i$ gets color $i$ and $y_i$ gets color $i'$; applying $\pi$ then yields a 3-coloring of $G'$, while $z_{i,e}$, whose only available colors are $i$ and $i'$, forbids the endpoints of $e$ from using $i$ and $i'$ simultaneously.

    \begin{claim}
        If $G'$ has a proper 3-coloring extending the given precoloring, then $G$ has a proper 6-coloring.
    \end{claim}
    \begin{claimproof}
        Let $\phi$ be such a coloring of $G'$, so $\phi(x_i)=\phi(y_i)=i$ for $i \in [3]$.
        Color $G$ by setting $c(v) = \phi(v)$ for $v \in X'$ and $c(v)=\phi(v)'$ for $v \in Y'$, and, for $i \in [3]$ and $e = xy \in E(G')$ with $x \in X'$ and $y \in Y'$, setting $c(z_{i,e}) = i$ if $\phi(x) \neq i$ and $c(z_{i,e})=i'$ otherwise; note that in the latter case $\phi(y) \neq i$, as $\phi(x)=i$ and $\phi$ is proper.

        In particular, $c(x_i)=i$ and $c(y_i)=i'$, so the six vertices of $Q$ receive six distinct colors, and the vertices of $X'$ receive unprimed colors while those of $Y'$ receive primed ones.
        Hence all edges within $V(G')$ are properly colored.
        Each $Z_i$ is independent and receives colors from $\{i,i'\}$, so the edges within $Z$ are properly colored as well.
        Finally, consider $z_{i,e}$ for $e=xy$ as above and recall that $c(z_{i,e}) \in \{i,i'\}$, while the vertices of $Q \setminus \{x_i,y_i\}$ receive the remaining four colors.
        If $c(z_{i,e})=i$, then $c(x) = \phi(x) \neq i$ and $c(y)$ is primed; if $c(z_{i,e})=i'$, then $c(x)$ is unprimed and $c(y) = \phi(y)' \neq i'$.
        Thus $c$ is proper.
    \end{claimproof}

    \begin{claim}
        If $G$ has a proper 6-coloring, then $G'$ has a proper 3-coloring extending the given precoloring.
    \end{claim}
    \begin{claimproof}
        Let $c$ be a proper 6-coloring of $G$.
        As $Q$ is a clique with six vertices, all six colors appear on $Q$, so, renaming the colors if necessary, we may assume that $c(x_i)=i$ and $c(y_i)=i'$ for $i \in [3]$.
        Define $\phi(v) = \pi(c(v))$ for $v \in V(G')$; clearly $\phi(x_i)=\phi(y_i)=i$, so $\phi$ extends the precoloring.

        Suppose $\phi$ is not proper, i.e., $\phi(x)=\phi(y)=i$ for some $e = xy \in E(G')$ with $x \in X'$ and $y \in Y'$.
        Then $c(x),c(y) \in \{i,i'\}$ and, as $X'$ and $Y'$ are complete to each other in $G$, we have $c(x) \neq c(y)$ and thus $\{c(x),c(y)\}=\{i,i'\}$.
        But the neighborhood of $z_{i,e}$ contains $\{x,y\} \cup (Q \setminus \{x_i,y_i\})$ and thus all six colors, a contradiction.
    \end{claimproof}

    We are left with proving that $G$ is $(P_4+P_2)$-free.
    Let us record four properties of $G$; the first three follow directly from the construction.
    \begin{enumerate}[(P1)]
        \item\label{p:bip} $G[X \cup Y]$ is complete bipartite, $X$ is complete to $\{y_1,y_2,y_3\}$ and anticomplete to $\{x_1,x_2,x_3\}$, and symmetrically for $Y$.
        In particular, $Q \subseteq N(x) \cup N(y)$ for every $x \in X$ and $y \in Y$.
        \item\label{p:multi} $G[Z]$ is complete multipartite with parts $Z_1,Z_2,Z_3$; in particular, if $z \in Z_i$ and $z' \in Z_j$ for $i \neq j$, then every vertex of $Z$ is adjacent to $z$ or to $z'$.
        \item\label{p:domQ} $Q \setminus N(z) \subseteq \{x_i,y_i\}$ for every $z \in Z_i$; in particular, $Q \subseteq N(z) \cup N(z')$ if moreover $z' \in Z_j$ for some $j \neq i$.
        \item\label{p:star} For every $i \in [3]$, the graphs $G[X \cup \{x_i\} \cup Z_i]$ and $G[Y \cup \{y_i\} \cup Z_i]$ are $P_4$-free.
        Indeed, consider the former one: both $X \cup \{x_i\}$ and $Z_i$ are independent, and the only possible neighbor of $z_{i,e}$ in $X \cup \{x_i\}$ is the endpoint of $e$ in $X'$.
        As the two internal vertices of an induced $P_4$ are adjacent and each has two neighbors in the path, one of them would be a vertex of $Z_i$ with two neighbors in $X \cup \{x_i\}$, a contradiction.
    \end{enumerate}
    Let us also note that complete multipartite graphs are $P_4$-free, so by \ref{p:bip} and \ref{p:multi} neither $G[X \cup Y]$ nor $G[Z]$ contains an induced $P_4$.

    \begin{claim}
        $G$ is $(P_4+P_2)$-free.
    \end{claim}
    \begin{claimproof}
        Suppose that $A,B \subseteq V(G)$ are disjoint and anticomplete to each other, and such that $G[A]$ is a $P_4$ and $G[B]$ is a $P_2$, say $B = \{u,v\}$.
        As $Q$ is a clique, it intersects at most one of $A,B$.
        
        \noindent\textbf{Case 1.} $A \cap Q \neq \emptyset$.
            Then $u,v \notin Q$ and thus, as $X$, $Y$, and each $Z_i$ are independent sets, up to swapping $u$ and $v$ we have $u \in X$ and $v \in Y$, or $u \in Z_i$ and $v \in Z_j$ for some $i \neq j$, or $u \in Z_i$ and $v \in X \cup Y$.
            In the first two cases $Q \subseteq N(u) \cup N(v)$ by \ref{p:bip} and \ref{p:domQ}, respectively, contradicting $A \cap Q \neq \emptyset$.
            So the last case holds and, by the symmetry between $X'$ and $Y'$, we may assume that $v \in X$.
            By \ref{p:domQ} and \ref{p:bip} we have $Q \setminus (N(u) \cup N(v)) \subseteq \{x_i,y_i\} \setminus \{y_i\}$, so $A \cap Q = \{x_i\}$.
            Furthermore, $A \cap Y = \emptyset$ as $v$ is complete to $Y$, and $A \cap Z \subseteq Z_i$ by \ref{p:multi}.
            Thus $A \subseteq X \cup \{x_i\} \cup Z_i$, contradicting \ref{p:star}.

        \noindent\textbf{Case 2.} $A \cap Q = \emptyset$.
            Suppose first that $A$ contains some $z \in Z_i$ and $z' \in Z_j$ with $i \neq j$.
            Then $B$ is disjoint from $Q$ by \ref{p:domQ} and from $Z$ by \ref{p:multi}, so, up to swapping $u$ and $v$, we have $u \in X$ and $v \in Y$.
            But then $A \cap (X \cup Y) = \emptyset$ and thus $A \subseteq Z$, which is impossible as $G[Z]$ is $P_4$-free.

            Hence $A \cap Z \subseteq Z_i$ for some $i \in [3]$.
            If $A$ contains some $x \in X$ and some $y \in Y$, then $B$ is disjoint from $Q$ by \ref{p:bip} and from $X \cup Y$, as $x$ is complete to $Y$ and $y$ is complete to $X$.
            So $u \in Z_k$ and $v \in Z_\ell$ for some $k \neq \ell$, and thus $A \cap Z = \emptyset$ by \ref{p:multi}.
            This gives $A \subseteq X \cup Y$, again contradicting the $P_4$-freeness of $G[X \cup Y]$.
            So $A \subseteq X \cup Z_i$ or $A \subseteq Y \cup Z_i$, contradicting \ref{p:star}.
    \end{claimproof}

    This completes the proof of \cref{thm:p4p2hard}.
\end{proof}

\newpage
\section{\col{$3$} in $P_t$-free graphs}
In this section, we prove \cref{thm:pthard}.
The crucial ingredient is the construction of another \emph{selector gadget}.
Let $\ell \geq 1$ be an integer.
Consider a tuple $\cS=(S,L,\mathbf{v})$, where $S$ is a graph, $L : V(S) \to 2^{[3]}$ is a list function, and $\mathbf{v} = (v_1,\ldots,v_\ell)$ is a sequence of distinct vertices of $S$ called \emph{ports}.
For $\iota \in \{2,3\}$, we say that a proper coloring $\phi$ of $S$ that respects lists $L$ \emph{$(1,\iota)$-selects} $i \in [\ell]$ if $\phi(v_i)=1$ and $\phi(v_j) = \iota$ for all $j \in [\ell]\setminus \{i\}$.
A tuple $\cS$ is called a \emph{$(1,\iota)$-selector} if:
\begin{enumerate}
    \item for every $i \in [\ell]$, there exists a proper coloring $\phi_i$ of $S$ that respects lists and $(1,\iota)$-selects $i$, and
    \item for every proper coloring $\phi$ of $S$ that respects lists, there exists $i \in [\ell]$ such that $\phi$ $(1,\iota)$-selects $i$.
\end{enumerate}
If $\iota$ is clear from the context, we will simply refer to a $(1,\iota)$-selector as a \emph{selector}.

\begin{lemma}\label{lem:selector}
    For $\iota \in \{2,3\}$ and every $\ell \geq 1$, in polynomial time one can construct a $(1,\iota)$-selector $\cS=(S,L,\mathbf{v})$ where $S$ is $P_{13}$-free and has $\Oh(\ell)$ vertices.
\end{lemma}
\begin{proof}
    By symmetry, we may assume that $\iota=2$.
    If $\ell=1$, then a single-vertex graph with a list $\{1\}$ is a selector.
    So, assume that $\ell \geq 2$.

    We will build the gadget in several steps.
    At each step, we work with an induced subgraph of the final graph and analyze
    its possible list colorings.

    \paragraph{Step 1.}
    For each $i \in [\ell]$, we introduce a triangle $T_i$ with vertices $a_i,b_i,c_i$.
    We set $L(a_i)=\{1,2\}$ and $L(c_i)=\{2,3\}$ for every $i \in [\ell]$,
    set $L(b_i)=\{1,3\}$ for every $i \in [\ell-1]$, and set $L(b_\ell)=\{3\}$.
    Furthermore, we add an edge $a_ib_j$ whenever $i < j$.

    For every $i \in [\ell-1]$, there are two possible list colorings of $T_i$:
    $(a_i,b_i,c_i)$ is colored either $(1,3,2)$ or $(2,1,3)$.
    We call them, respectively, \emph{type 1} and \emph{type 2}.
    The triangle $T_\ell$ is forced to have a coloring of type~1.

    Note that if $i < j$ and $T_i$ has coloring of type 1, then $T_j$ cannot have coloring of type 2 because of the edge $a_ib_j$.
    Consequently, in every proper coloring of the graph that we constructed so far,
    there exists $t \in [\ell]$ such that $T_i$ has coloring of type 2 whenever $i < t$ and type 1 whenever $i \geq t$.
    Moreover, such a coloring exists for every $t \in [\ell]$.

    \paragraph{Step 2.}
    For each $i \in [\ell-1]$, we add three vertices $d_i,e_i,f_i$ and edges $a_ie_i$, $a_if_i$, $d_ie_i$, and $d_if_i$.
    We set $L(d_i)=L(e_i)=\{1,3\}$ and $L(f_i)=\{1,2\}$.
    If $T_i$ has coloring of type 1, then $(d_i,e_i,f_i)$ is colored $(1,3,2)$,
    and if $T_i$ has coloring of type 2, then $(d_i,e_i,f_i)$ is colored $(3,1,1)$.

    Consequently, in every proper coloring $\phi$,
    there exists $t \in [\ell]$ such that
\begin{equation}
\label{eq:coloringstep2}
\begin{aligned}
    \phi(d_i)&=\begin{cases}
        3 & \text{ if } i < t,\\
        1 & \text{ if } i \geq t.
    \end{cases}
    && \text{for every } i \in [\ell-1], \qquad
    \phi(b_i)&=\begin{cases}
        1 & \text{ if } i < t,\\
        3 & \text{ if } i \geq t.
    \end{cases}
    && \text{for every } i \in [\ell].
\end{aligned}
\end{equation}

\paragraph{Step 3.}
    Next, we add a vertex $d_0$ with list $\{3\}$.
    For each $i \in [\ell]$, we define $x_i=d_{i-1}$ and $y_i=b_i$.

    Now, \eqref{eq:coloringstep2} yields that in every proper coloring $\phi$,
    there exists $t \in [\ell]$ such that
\begin{equation}
\label{eq:coloringstep3}
    (\phi(x_i),\phi(y_i)) =
    \begin{cases}
        (3,1) & \text{ if } i < t, \\
        (3,3) & \text{ if } i = t, \\
        (1,3) & \text{ if } i > t.
    \end{cases}
\end{equation}
    Moreover, for every $t \in [\ell]$ the corresponding coloring exists.

\paragraph{Step 4.}

Finally, for each $i \in [\ell]$, we add the port vertex $v_i$ and two additional vertices $u_i,w_i$,
and edges $v_ix_i$, $v_iy_i$, $w_iv_i$, $w_ix_i$, $u_iw_i$, and $y_iu_i$.
We set lists $L(v_i) = \{1,2\}$, $L(u_i) = \{1,3\}$, and $L(w_i) = \{1,2,3\}$.
For each of the possible ordered pairs $(\phi(x_i),\phi(y_i))$, the unique extension to $\{v_i,u_i,w_i\}$ is as follows:
\begin{equation}
\label{eq:coloringstep4}
\begin{array}{cc|ccc}
    x_i & y_i & v_i & u_i & w_i \\ \hline
    3 & 1 & 2 & 3 & 1 \\
    3 & 3 & 1 & 1 & 2 \\
    1 & 3 & 2 & 1 & 3
\end{array}
\end{equation}
If $x_i$ and $y_i$ are both colored 1, then there is no way to color the subgraph properly.

This completes the construction of the selector $\cS=(S,L,\mathbf{v})$;
note that $S$ has $9\ell - 2 = \Oh(\ell)$ vertices and can clearly be constructed in polynomial time.
Combining \eqref{eq:coloringstep3} and \eqref{eq:coloringstep4}, we obtain the following claim.

\begin{claim}
    In every proper coloring $\phi$ of the gadget, there exists $t \in [\ell]$ such that $\phi(v_t) = 1$ and $\phi(v_j) = 2$ for all $j \in [\ell]\setminus \{t\}$.
    Furthermore, for every $t \in [\ell]$ the corresponding coloring exists.
\end{claim}

Thus, $(S,L,\mathbf{v})$ is indeed a selector.

\begin{claim}
    $S$ is $P_{13}$-free.
\end{claim}
\begin{claimproof}
    Define $K=\bigcup_{i \in [\ell]} \{a_i,b_i\}$.
    Note that $S[K]$ is a chain graph, i.e., it is $2P_2$-free, so in particular, $P_5$-free.
    There are three possible types of components of $S-K$:
    \begin{itemize}
        \item a single vertex $c_i$ with neighborhood $\{a_i,b_i\}$,
        \item $\{d_0,v_1,u_1,w_1\}$ with neighborhood $\{b_1\}$, and
        \item $\{d_i,e_i,f_i,v_{i+1},u_{i+1},w_{i+1}\}$ for every $i \in [\ell-1]$, with neighborhood $\{a_i,b_{i+1}\}$.
    \end{itemize}
    Direct verification shows that each of these components is $P_5$-free, and that the neighborhood of every component in $K$ is either a single vertex, or two vertices joined with an edge.

    Let $P$ be a longest induced path in $S$.
    Note that it is impossible that $P$ visits $K$, leaves this set, and then visits $K$ again, because the neighborhood of every component of $S-K$ in $K$ is a clique.
    Thus, the path $P$ can be split into (at most) three subpaths: the first and the last one are contained in a component of $S-K$, and the middle one is contained in $K$.
    The first and the last subpath have at most 4 vertices each, and the middle one also has at most 4 vertices.
    Consequently, $P$ has at most $4+4+4=12$ vertices. This means that $S$ is $P_{13}$-free.
\end{claimproof}

   This completes the proof of \cref{lem:selector}.
\end{proof}

Before we proceed to the hardness proof, let us introduce an auxiliary lemma.

\begin{lemma}\label{lem:combinept}
    Let $G$ be a graph, and $B_1,\ldots,B_p$ be a partition of its vertex set, such that:
    \begin{itemize}
    \item for every $i \in [p]$, the graph $G[B_i]$ is $P_a$-free, where $a \geq 2$, and
    \item for every $i,j \in [p]$, where $i < j$, the largest induced matching in the graph $G[B_i, B_j]$ has $b_{i,j}$ edges.
    \end{itemize}
    Then, $G$ is $P_{c}$-free, where $c = (a-1)(1 + 3 \sum_{1 \leq i < j \leq p} b_{i,j})+1$.
\end{lemma}
\begin{proof}
    An edge of $G$ is called \emph{crossing} if its endpoints belong to different sets $B_i$ and $B_j$.
    Let $P$ be an induced path in $G$, let $q$ be the number of its crossing edges, and let $s=\sum_{1 \leq i < j \leq p} b_{i,j}$.
    Enumerate the consecutive edges of $P$ as $e_1,\ldots,e_m$.
    For $r \in \{0,1,2\}$, let $M_r=\{e_i \mid i \equiv r \pmod 3\}$.
    The sets $M_0,M_1,M_2$ partition $E(P)$.
    Since $P$ is induced, every $M_r$ is an induced matching in $G$.

    Fix $i<j$ in $[p]$.
    For every $r \in \{0,1,2\}$, the set $M_r \cap E(B_i,B_j)$ is an induced matching in $G[B_i,B_j]$.
    Hence,
    \[
        |E(P) \cap E(B_i,B_j)| = \sum_{r=0}^2 |M_r \cap E(B_i,B_j)| \leq 3b_{i,j}.
    \]
    Consequently, $q \leq 3s$.

    Removing the crossing edges from $P$ splits it into $q+1$ subpaths, each of which is contained in a single set $B_i$.
    Since $G[B_i]$ is $P_a$-free, each such subpath has at most $a-1$ vertices and hence at most $a-2$ edges.
    Therefore,
    \[
        |E(P)| \leq q+(q+1)(a-2) = (a-1)(q+1)-1 \leq (a-1)(1+3s)-1.
    \]
    Thus, $P$ has at most $(a-1)(1+3s)$ vertices.
    As this holds for every induced path $P$ in $G$, the graph $G$ is $P_c$-free.
\end{proof}

\begin{theorem}
    \label{thm:listpt}
    \lcol{$3$} on $P_t$-free graphs is \Wone-hard when parameterized by $t$.
    Furthermore, the problem on $n$-vertex instances cannot be solved in time $f(t) \cdot n^{o(t/\log t)}$, for any computable function $f$, unless the ETH fails.
\end{theorem}
\begin{proof}
    We reduce from \MSI with subcubic pattern graph.
    Let $(G',H)$ be an instance, and denote the vertices of the pattern graph $H$ by $[r]$.
    Since $H$ is subcubic, it has at most $3r/2$ edges.
    We use the convention that if we refer to an edge $ij$ of $H$, then $i<j$.
    Set $t=216r+13$.

    Let $V^1,\ldots,V^r$ be the partition of the vertex set of the host graph $G'$ into independent sets,
    where the set $V^i$ corresponds to the vertex $i$ of $H$. Assume that $|V^1|=\ldots=|V^r|=\ell$.
    For an edge $ij$ of $H$, let $E^{ij}$ be the set of edges between $V^i$ and $V^j$ in $G'$ and let $\ell_{ij}$ denote $|E^{ij}|$. We may assume that, for each $ij \in E(H)$, it holds that $\ell_{ij} >0$, as otherwise we are clearly dealing with a no-instance of \MSI.
    Again, we use the convention that if we refer to an edge $uv \in E^{ij}$, then $u \in V^i$ and $v \in V^j$.

    Fix arbitrary orderings of every set $V^i$ and every set $E^{ij}$.
    We take all selectors introduced below to be vertex-disjoint, and use $L$ for the union of their list functions.
    For each $i \in [r]$, we introduce a $(1,2)$-selector $\cS^i=(S^i,L,\mathbf{v}^i)$ with $\ell$ ports, denoted by $\mathbf{v}^i = ( \port(v) )_{v \in V^i}$.
    For every edge $ij$ of $H$, we introduce a $(1,3)$-selector $\cS^{ij}=(S^{ij},L,\mathbf{w}^{ij})$ with $\ell_{ij}$ ports denoted by $\mathbf{w}^{ij}=( \port(e) )_{e \in E^{ij}}$.

    The coloring of $S^i$ (resp., $S^{ij}$) will represent the choice of a vertex from $V^i$ (resp., edge from $E^{ij}$) in the solution to \MSI.
    We need to make sure that the choice of edges is consistent with the choice of vertices, i.e., if we choose an edge $ab \in E^{ij}$, then we must also choose the vertices $a \in V^i$ and $b \in V^j$.

    To this end, for every edge $ij$ of $H$,
    and every edge $ab \in E^{ij}$,
    we add edges $\port(ab)\port(a')$ for every $a' \in V^i \setminus \{a\}$
    and edges $\port(ab)\port(b')$ for every $b' \in V^j \setminus \{b\}$.

    This completes the construction of the graph $G$ and the list function $L$.
    The construction can be performed in polynomial time.
    Moreover,
    \[
        |V(G)|=\Oh\left(r\ell + \sum_{ij \in E(H)} \ell_{ij}\right)=\Oh(r\ell^2)=\Oh(|V(G')|^2).
    \]

    In the next two claims we show that $(G',H)$ is a yes-instance of \MSI if and only if $(G,L)$ is a yes-instance of \lcol{$3$}.

    \begin{claim}
        If $(G',H)$ is a yes-instance of \MSI, then $(G,L)$ is a yes-instance of \lcol{$3$}.
    \end{claim}
    \begin{claimproof}
        Let $(a_i)_{i \in [r]}$ be a solution to the instance $(G',H)$ of \MSI.
        For every $i \in [r]$, color the selector $S^i$ so that it $(1,2)$-selects $a_i$.
        For every edge $ij$ of $H$, use the coloring of $S^{ij}$ in which the port $\port(a_ia_j)$ has color $1$ and every other port has color $3$.

        In the resulting coloring, the port $\port(a_i)$ has color $1$ and every other port of $S^i$ has color $2$.
        Similarly, the port $\port(a_ia_j)$ has color $1$ and every other port of $S^{ij}$ has color $3$.
        It remains to verify that the edges added between selectors are properly colored.
        Consider such an edge incident with $\port(ab) \in S^{ij}$.
        If $ab\neq a_ia_j$, then $\port(ab)$ has color $3$, whereas its neighbor in $S^i$ or $S^j$ has color $1$ or $2$.
        If $ab=a_ia_j$, then every neighbor of $\port(ab)$ is a port corresponding to a vertex different from $a_i$ or $a_j$, respectively, and hence has color $2$.
        Thus, the colorings of all selectors together form a proper coloring of $G$ that respects lists.
    \end{claimproof}

    \begin{claim}
        If $(G,L)$ is a yes-instance of \lcol{$3$}, then $(G',H)$ is a yes-instance of \MSI.
    \end{claim}
    \begin{claimproof}
        Let $\phi$ be a proper coloring of $G$ that respects lists.
        For every $i \in [r]$, the restriction of $\phi$ to the selector $S^i$ selects a unique vertex $a_i \in V^i$.
        For every edge $ij$ of $H$, there is a unique edge $ab \in E^{ij}$, where $a \in V^i$ and $b \in V^j$, such that $\phi(\port(ab))=1$; every other port of $S^{ij}$ has color $3$.

        The vertices $\port(ab)$ and $\port(a')$ are adjacent for every $a' \in V^i\setminus\{a\}$.
        Since both $\port(ab)$ and $\port(a_i)$ have color $1$, properness of $\phi$ implies that $a_i=a$.
        The symmetric argument gives $a_j=b$.
        Hence, $a_ia_j=ab \in E^{ij}$ for every edge $ij$ of $H$, and $(a_i)_{i \in [r]}$ is a solution to $(G',H)$.
    \end{claimproof}

    Now, it remains to show that $G$ is $P_t$-free for $t = 216r+13$.
    \begin{claim}
        $G$ is $P_t$-free.
    \end{claim}
    \begin{claimproof}
        Let $\mathcal{B}$ be the partition of $V(G)$ whose parts are the vertex sets of all selectors.
        By \cref{lem:selector}, every part of $\mathcal{B}$ induces a $P_{13}$-free graph.

        By construction, edges between distinct parts of $\mathcal{B}$ occur only between $S^{ij}$ and $S^i$, or between $S^{ij}$ and $S^j$, for an edge $ij$ of $H$.
        We claim that each of these bipartite graphs has no induced matching with three edges.
        Consider the graph between $S^{ij}$ and $S^i$; the other case is symmetric.
        Every edge in this graph has the form $\port(a_\alpha b_\alpha)\port(x_\alpha)$, where $a_\alpha b_\alpha \in E^{ij}$ and $x_\alpha \in V^i\setminus\{a_\alpha\}$.
        If three such edges formed an induced matching, then for any distinct $\alpha,\beta \in [3]$, the non-adjacency of $\port(a_\alpha b_\alpha)$ and $\port(x_\beta)$ would imply $x_\beta=a_\alpha$, by the definition of the edges between the two selectors.
        Taking $\alpha=1$ and $\beta\in\{2,3\}$ gives $x_2=x_3=a_1$, contradicting the fact that the three edges form a matching.
        Thus, the largest induced matching between any two distinct parts of $\mathcal{B}$ has at most two edges.

        There are at most $2|E(H)| \leq 3r$ pairs of parts that are adjacent, and all other pairs have no edges between them.
        Therefore, the sum of the largest induced-matching sizes over all pairs of distinct parts of $\mathcal{B}$ is at most $6r$.
        Applying \cref{lem:combinept} with $a=13$ yields that $G$ is $P_{12(1+3\cdot 6r)+1}$-free, that is, $P_{216r+13}$-free.
    \end{claimproof}

    As $t=216r+13$ and $|V(G)| =\Oh(|V(G')|^2)$, both the \Wone-hardness and the ETH lower bound follow from \cref{thm:msi}.
\end{proof}

Finally, we can show \cref{thm:pthard}.
\thmpt*
\begin{proof}
    We start with an instance $(G',L)$ of \lcol{$3$} given by \cref{thm:listpt}, where $G'$ is $P_{t'}$-free.
    We apply the same construction as in the proof of \cref{thm:2p2hard}:
    we introduce a triangle $q_1,q_2,q_3$, and add an edge $q_iv$ for every $i \in [3]$ and every $v \in V(G')$ such that $i \notin L(v)$.
    Let $G$ be the resulting graph. Clearly, it admits a proper $3$-coloring if and only if $(G',L)$ is a yes-instance of \lcol{$3$}.

    We claim that $G$ is $P_t$-free, where $t = 2t' + 1$.
    Let $P$ be a longest induced path in $G$.
    It uses at most two vertices from $\{q_1,q_2,q_3\}$ and, if they are used, they must be consecutive on the path.
    Thus, $P$ can be split into at most three subpaths: the first and the last one are contained in $G'$, and the middle uses at most 2 vertices among $\{q_1,q_2,q_3\}$.
    As each of the subpaths contained in $G'$ has at most $t'-1$ vertices,
    we conclude that $P$ has at most $2(t'-1) + 2 = 2t'$ vertices.
    Thus, $G$ is indeed $P_{2t'+1}$-free.

    As $t = 2t' + 1$ and $n = |V(G)| = |V(G')| + 3$, both the \Wone-hardness and the ETH lower bound from \cref{thm:listpt} transfer to \col{$3$}.
\end{proof}

\newpage
\section{Vertex-critical $(P_4+sP_1)$-free graphs}
Now, let us show that our approach from the proof of \cref{thm:p4sp1algo} can be used to bound the number of vertex-$k$-critical $(P_4+sP_1)$-free graphs and, in particular, to prove \cref{thm:obstructions}.

\thmobstructions*

Let us first dispose of the degenerate case $s=0$.
Note that $P_4$-free graphs are perfect and thus they are $k$-colorable if and only if they do not contain a clique of size $k+1$.
Consequently, $K_{k+1}$ is the only vertex-$(k+1)$-critical cograph, so from now on we may assume that $s \geq 1$.

Actually, we will prove a stronger statement, concerning list $k$-colorings.
We say that an instance $(G,L)$ is a \emph{minimal obstruction to list $k$-colorability} if $(G,L)$ is a no-instance, but, for every $v \in V(G)$, $(G-v,L)$ is a yes-instance of \lcol{$k$}.

\begin{theorem}\label{thm:min-list-obstructions}
    Let $s,k \geq 1$.
    Every $(P_4+sP_1)$-free minimal obstruction to list $k$-colorability has at most $kM$ vertices, where $M$ is as in \eqref{eq:constants}.
    In particular, treating $s$ as a constant, this bound is $2^{\Oh(k^2)}$.
\end{theorem}
\begin{proof}
    Let $(G,L)$ be a $(P_4+sP_1)$-free minimal obstruction to list $k$-colorability.
    For contradiction, suppose that $|V(G)| > kM$.
    We apply \cref{cor:reduce} to $(G,L)$ and discuss its three possible outcomes.

    Suppose that the first outcome is returned, i.e., a largest independent set of $G$ has size at most $M-1$.
    Let $v$ be an arbitrary vertex of $G$. By minimality, $(G-v,L)$ is a yes-instance of \lcol{$k$}, so in particular, $G-v$ is $k$-colorable.
    This means that $|V(G)| \leq k(M-1)+1 \leq kM$, a contradiction.

    If the second outcome, i.e., a clique of size $k+1$ is returned, then this clique itself is an obstruction to list $k$-colorability. As $k+1 < kM$, we have a contradiction with minimality of $(G,L)$.

    Finally, if the third outcome is returned, then we obtain a proper induced subgraph $G'$ of $G$ such that $(G',L)$ is a no-instance of \lcol{$k$}. Again, this contradicts the minimality of $(G,L)$.
\end{proof}

A vertex-$(k+1)$-critical graph, equipped with lists all equal to $[k]$, is a minimal obstruction to list $k$-colorability.
Thus \cref{thm:min-list-obstructions} implies that every $(P_4+sP_1)$-free vertex-$(k+1)$-critical graph has at most $kM = 2^{\Oh(k^2)}$ vertices, and so there are only finitely many such graphs.
Together with the case $s=0$ discussed above, this completes the proof of \cref{thm:obstructions}.

For the non-list setting, we can slightly improve the bound on the number of vertices in a vertex-$(k+1)$-critical $(P_4+sP_1)$-free graph.

\begin{theorem}\label{thm:min-obstructions}
    Let $s,k \geq 1$ be integers and define
    \[
        q' = \max (2s+3,sk+2), \quad p' = sq', \quad M'= h_s(2k,p').
    \]
    Every $(P_4+sP_1)$-free vertex-$(k+1)$-critical graph has at most $kM'$ vertices.
    In particular, treating $s$ as a constant, this bound is $2^{\Oh(k \log k)}$.
\end{theorem}
\begin{proof}
    The proof outline is essentially the same as for the list version, but we use the parameters $q',p',M'$ instead of $q,\sigma,M$.
    Let us only sketch the necessary modifications; as we are interested solely in the bound on the number of vertices, we ignore all algorithmic aspects.
    Let $G$ be a $(P_4+sP_1)$-free vertex-$(k+1)$-critical graph, i.e., $G$ is not $k$-colorable, but every proper induced subgraph of $G$ is.
    Suppose that $|V(G)| > kM'$.

    If every independent set of $G$ has fewer than $M'$ vertices, then, as $G-v$ is $k$-colorable for any $v \in V(G)$, we obtain $|V(G)| \leq k(M'-1)+1 \leq kM'$, a contradiction.
    So let $X$ be an independent set of size $M' = h_s(2k,p')$ and apply \cref{lem:extract-is} to $G$ and $X$, with $p := p'$.
    If a clique $Q$ of size $k+1$ is returned, then $Q$ is not $k$-colorable, so $G = Q$ by criticality, and thus $|V(G)| = k+1 \leq kM'$, again a contradiction.
    Otherwise, we obtain a set $S$ of size $p'$ satisfying \eqref{eq:extract-is}, and then \cref{lem:components}, applied with $q := q' \geq 2s+3$ and $\sigma := p' = sq'$, gives us a partition $W,D$ of $V(G)$ and distinguished components $C_1,\ldots,C_{q'}$ of $G[W]$.

    Finally, we modify \cref{lem:reduce}; here we use the fact that all lists are equal to $[k]$.
    Indeed, for $K \subseteq [k]$, the instance $(G[C_i],L|K)$ is a no-instance of \lcol{$k$} if and only if $\chi(G[C_i]) > |K|$, so it depends on $K$ only through $|K|$, and it is monotone in $|K|$.
    Consequently, instead of the $2^k$ families $\cC_K$, it suffices to consider one family $\cC$, consisting of $sk+1$ distinguished components with the largest chromatic number;
    as $q' \geq sk+2$, the union $Z$ of the remaining distinguished components is nonempty.
    In the proof of \cref{clm:deleted-components} we now argue as follows: if a deleted component $C$ satisfies $\chi(G[C]) > |K|$, then $\chi(G[C_i]) \geq \chi(G[C]) > |K|$ for every $C_i \in \cC$, so on each of the $sk+1$ surviving components in $\cC$ the coloring $\phi$ must use a color outside $K$, i.e., a color from $A$; this contradicts \cref{clm:surviving-components}.
    The remaining arguments are unchanged, so $G$ is $k$-colorable if and only if $G-Z$ is.
    As $Z \neq \emptyset$ and $G-Z$ is a proper induced subgraph of $G$, the graph $G-Z$ is $k$-colorable, and thus so is $G$, a contradiction.

    It remains to estimate $M'$. As $s$ is a constant, we have $q' = \Oh(k)$ and $p' = \Oh(k)$, so
    \[
        M' = p'\left(1+2k \cdot (p')^s\right)^{2k} = k^{\Oh(k)} = 2^{\Oh(k \log k)},
    \]
    and hence $kM' = 2^{\Oh(k \log k)}$ as well.
\end{proof}

\paragraph{Acknowledgement.} ChatGPT and Claude were used to simplify some constructions and to identify or correct omissions and minor mistakes.
 The author takes full responsibility for the content of the paper.

\bibliographystyle{abbrv}
\bibliography{main}
\end{document}